%% file: Main_arxiv_v1.tex
\documentclass{article}

\input{Packages}
\input{Commandes}

\usepackage[normalem]{ulem}

\title{\bf{Large $N$ factorization of families of tensor trace-invariants}}
\author[1]{Sylvain Carrozza\thanks{sylvain.carrozza@ube.fr}}
\author[1,2]{Johann Chevrier\thanks{johann.chevrier@ube.fr}}
\author[2]{Luca Lionni\thanks{luca.lionni@ens-lyon.fr}}
\affil[1]{\normalsize{Universit\'{e} Bourgogne Europe, CNRS, IMB UMR 5584, 21000 Dijon, France}}
\affil[2]{\normalsize{CNRS, ENS de Lyon, LPENSL, UMR5672, 69342 Lyon cedex 07, France}}
\date{\today}

\begin{document}

\maketitle

\begin{abstract}
\noindent It was recently proven that, in contrast to their matrix analogues, the moments of a real Gaussian tensor of size $N$ do not in general factorize over their connected components in the asymptotic large $N$ limit. While the original proof of this rather surprising result was not constructive, explicit examples of non-factorizing moments, which are expectation values of trace-invariants, have since then been discovered. We explore further aspects of this problem, with a focus on Haar-distributed (or Gaussian) \emph{complex} random tensors, which are more directly relevant to quantum information. \\ We start out by exhibiting an explicit example of non-factorizing trace-invariant, thereby filling a gap in the recent literature. We then turn to the opposite question: that of finding interesting families of trace-invariants that \emph{do} in fact factorize at large $N$. We establish three main theorems in this regard. The first one provides a sufficient combinatorial bound ensuring large $N$ factorization, that is also simple enough to be applicable to various cases of practical relevance. Our second main result shows that the expectation value of any \emph{compatible} trace-invariant is dominated by certain tree-like combinatorial structures at large $N$, which we refer to as \emph{tree-like dominant pairings}. Our third main theorem establishes that any trace-invariant admitting tree-like dominant pairings does actually factorize at large $N$. In this way, we are able to prove that various families of trace-invariants that have been previously studied in the literature do factorize at large $N$.\\ We apply our findings to the theory of multipartite quantum entanglement: to any trace-invariant is associated a multipartite generalization of R\'{e}nyi entanglement entropy, whose typical expectation value in the uniform random quantum state can be explicitly computed assuming large $N$ factorization.
\end{abstract}

\setcounter{tocdepth}{2}
\tableofcontents

\section{Introduction}

The present contribution is a companion paper to Ref.~\cite{Carrozza:2026qcf}. Our main purpose here is to investigate the consequences of the fundamental large $N$ non-factorization result established in Ref.~\cite{Gurau:2025evo}, and further studied in Ref.~\cite{Berthold:2026zxk}. In contrast to those prior works, we will work in a complex setting (rather than a real one): since complex random tensors can be used to model random multipartite quantum states, this makes the present work directly relevant to Ref.~\cite{Carrozza:2026qcf}, as advertised.

In more detail, the random tensors considered are arrays of complex numbers $S=\{S_{i_1, \ldots, i_D}\}$, where the $D\ge 2$ indices are ordered (no symmetry is \apriori assumed) and take value in $\{1, \ldots, N\}$. For $D=2$, $S$ is a random matrix. In particular, we will mainly discuss the two following examples in this paper (for any fixed number of indices $D$): a) the complex Gaussian tensor $T$, whose $N^D$ components are independent and identically distributed complex Gaussian random variables with zero mean and variance $1/N^D$; and b) the Haar random tensor $T_U$, obtained by applying a $U(N^D)$ Haar-distributed unitary matrix to a fixed vector of norm 1 in $\mathbb{C}^{N^D}$.\footnote{Equivalently, $T_U$ can be defined as a normalized Gaussian tensor $T/\norm{T}$.} In both cases, the resulting random vector is interpreted as a random $D$-index tensor after subdividing its index set according to the isomorphism $\mathbb{C}^{N^D}\simeq (\mathbb{C}^{N})^{\otimes D}$.

A \emph{local unitary} ($\LU$) transformation on $(\mathbb{C}^{N})^{\otimes D}$ is a $D$-tuple of unitary transformations $\vec{U}=(U_1 , \ldots , U_D)$ on $\mathbb{C}^N$. Given a random tensor $S$, $S_{\vec{U}}$ denotes the random tensor obtained from $S$ by applying the $N\times N$ unitary matrix $U_c$ to the $c^{\rm{th}}$ index of $S$, for every $c$. $S$ is then said to be local-unitary invariant ($\LU$-invariant) if, for any local unitary $\vec{U}$, $S$ and $S_{\vec U}$ have the same distribution. Such distributions are for instance important for applications to random geometry \cite{Gurau:2016cjo, PhysRevD.85.084037, Bonzom:2011zz, Bonzom:2012cu, Bonzom:2015axa, gurau_random_2017, Lionni2018, Gurau:2024nzv, Carrozza:2024gnh} and quantum information \cite{Page:1993df,  PhysRevLett.72.1148, Sanchez-Ruiz:1995bhf, Sen:1996ph, Hayden:2005sqo, Iizuka:2024pzm,  dartois_injective_2024, Carrozza:2026qcf, Collins:2024pip, bucdalché2026, marginals, RMTQI, Hayden:2016cfa, Kudler-Flam:2021efr,  Penington:2022dhr, cheng2024random}. Both distributions $T$ and $T_U$ are $\LU$-invariant. A polynomial $P$ in the coefficients of a tensor $S$ and of its complex conjugate $\bar S$ is $\LU$-invariant if $P(S, \bar S)=P(S_{\vec U}, \bar{S}_{\vec U})$ for any local unitary $\vec{U}$. The quantities that define the moments (or correlations) of a $\LU$-invariant random tensor $S$ are the averages (or expectations) $\langle P(S,\bar S)\rangle$, where $P$ belongs to a subset of $\LU$-invariant polynomials  called \emph{trace-invariants}. 

One is naturally interested in understanding tensor distributions in the limit where $N$ is large. In this context, it is important to characterize the large $N$  asymptotics of the  moments $\langle P(S,\bar S)\rangle$ of $\LU$-invariant random tensors. 

\paragraph{The factorization problem.} For random matrices, it is sufficient to characterize the large $N$ asymptotics of a subset of trace-invariants: those that do not themselves factorize as products of lower degree trace-invariants. The following \emph{large $N$ factorization property} is indeed known to hold for a large class of random matrix distributions, including the previously introduced Gaussian and Haar-distributed matrices:\footnote{For a ``classical'' example of complex LU-invariant random matrix, one indeed expects the covariance of two polynomials $P_1, P_2$ as above to scale as $N$ and $\langle P_1\rangle \langle P_2 \rangle$ to scale as $N^2$ (see \eg Ref.~\cite{Collins:2024pip} Sec.~3.4), which justifies Eq.~\eqref{eq:facto-intro} using Eq.~\eqref{eq:facto-intro-conn-corr}.}
\begin{equation}
\label{eq:facto-intro}
    \langle P_1(S,\bar S)P_2(S,\bar S) \rangle \underset{N \to \infty}{\sim}   \langle P_1(S,\bar S)\rangle \langle P_2(S,\bar S) \rangle \;.
\end{equation}
When such an equation holds, we say that $P_1$ and $P_2$ factorize at large $N$ for the random tensor or matrix $S$. 
The large $N$ factorization Eq.~\eqref{eq:facto-intro} can be related to the fact that \emph{connected correlations} (or \emph{cumulants}) are suppressed by powers of $1/N$ with respect to products of correlations. One indeed has: 
\begin{equation}
\label{eq:facto-intro-conn-corr}
    \langle P_1(S, \bar S)P_2(S, \bar S)\rangle =  \langle P_1(S, \bar S)\rangle \langle P_2(S, \bar S)\rangle + \langle P_1(S, \bar S),P_2(S, \bar S)\rangle_{\textrm{conn}}  \;,
\end{equation}
where if $x, y$ are complex random variables, $\langle x , y \rangle_{\textrm{conn}}$ is the covariance of $x, \bar y$.

More broadly, such factorization properties are commonly encountered in large $N$ field theories, where they significantly simplify the analysis of strong-coupling effects. In particular, in vector or matrix models (including gauge theories), they lead to a uniform $1/N$ suppression of the fluctuations of relevant observables which allows one to reinterpret the large $N$ regime as a semi-classical one, via the identification $\hbar \sim 1/N$ (see \eg\cite{Yaffe:1981vf, Rossi:1996hs, Lucini:2012gg, GarciaVera:2018kfu} and references therein). In random tensor models (see \cite{gurau_random_2017} and references therein) and tensor field theories (see \cite{Klebanov:2018fzb, Gurau:2019qag, Benedetti:2020seh} and references therein), the large $N$ factorization property of specific families of observables was also noticed and taken advantage of, most notably in the context of large $N$ limits dominated by so-called \emph{melonic graphs}, of which we should distinguish two rather distinct types: melonic limits which generate \emph{local} radiative corrections \cite{PhysRevD.85.084037} and thus behave similarly to vector models in the quantum setting (see \eg \cite{Benedetti:2018ghn}); melonic limits which, by means of \emph{bilocal} radiative corrections \cite{Gurau:2010ba, Bonzom:2011zz, Carrozza:2015adg, Ferrari:2017jgw, Gurau:2017qya, Benedetti:2017qxl, Carrozza:2018ewt, Carrozza:2021qos}, can reproduce the interesting strong-coupling effects of the Sachdev-Ye-Kitaev quantum-mechanical model \cite{Witten:2016iux, Klebanov:2016xxf}, and give rise to an interesting class of large $N$ QFTs in higher dimensions \cite{Benedetti:2017fmp, Giombi:2018qgp, Benedetti:2019eyl, Benedetti:2019ikb, Benedetti:2020yvb}. In Ref.~\cite{Bonzom:2012cu}, large $N$ factorization was even taken as an input in order to derive the existence and properties of large $N$ limits of (non-Gaussian) random tensors from the infinite tower of Schwinger-Dyson equations obeyed by their correlation functions. To our knowledge, no example of random tensor distribution or large $N$ field theory that breaks the generally expected large $N$ factorization property was known until recently.

\

Last year, Ref.~\cite{Gurau:2025evo} provided a probabilistic proof of the existence of trace-invariants for which the large-$N$ factorization property of Eq.~\eqref{eq:facto-intro} does not hold in the real Gaussian case (and their proof generalizes straightforwardly to the complex setting). Moreover, such non-factorizing trace-invariants are in some sense generic (see Ref.~\cite{Gurau:2025evo} for a precise statement). Extrapolating to the quantum setting, one may interpret this result as the statement that, in large $N$ tensor field theories, generic observables should not be expected to behave semi-classically in the large $N$ limit. In light of this remarkable fact, one would like to first produce explicit examples that do not factorize at large $N$ for the complex Gaussian tensor, and then, eventually, to precisely determine the subset of trace-invariants that do not factorize. Very recently, Ref.~\cite{Berthold:2026zxk} provided some explicit examples of polynomials that do not factorize at large $N$ for \emph{real} random tensors, but the corresponding polynomials do not correspond to trace-invariants of \emph{complex} $\LU$-invariant tensors, so, as of today, there is still no known counterexample to large $N$ factorization for the complex case. 

In \textbf{Sec.~\ref{sec:counterex}}, we present the first counterexample to large $N$ factorization for a complex random tensor. Unlike the previous counterexamples of Ref.~\cite{Berthold:2026zxk}, which are based on non-bipartite colored graphs and thus cannot support complex tensor invariants, our result overcomes this limitation. While adapted to the complex setting, the method we followed to generate an explicit counterexample is completely analogous to the one proposed by the authors of Ref.~\cite{Berthold:2026zxk}: indeed, just like their counterexamples, ours falls in the class of \emph{maximally single-trace} invariants previously introduced in Ref.~\cite{Ferrari:2017jgw}. Additionally, we investigate conditions under which all cumulants may contribute to the leading term of a given moment at large $N$.

\ 
 
Approaching the question of large $N$ factorization from a broader perspective, we may study under what conditions some given collection $\{P_1, \ldots, P_s\}$ of trace-invariants factorizes. In Ref.~\cite{Gurau:2025evo}, a combinatorial condition is given, that involves the polynomials $P_1, \ldots, P_s$ themselves, but also a much larger set of polynomials built from those elements. As a result, it is arguably very difficult to verify in general whether the collection $\{ P_1, \ldots, P_s \}$ satisfies the aforementioned condition. 

It is the latter point of view that we adopt in this paper and especially in \textbf{Sec.~\ref{sec:largeN}}. Indeed, we provide two sufficient conditions for a restricted family of trace-invariants $\{P_1, \ldots, P_s\}$ to factorize at large $N$ for complex Gaussian tensors or, equivalently, for  Haar random tensors. On the one hand, Thm.~\ref{main-theo1} ensures the factorization of invariants associated with graphs satisfying a certain combinatorial bound. In the process, it also provides the tools to show that the so-called \textit{compatible graphs} necessarily exhibit a \textit{tree-like behavior at leading order} (see Def.~\ref{def:DEF} and Thm.~\ref{main-theo2}). On the other hand, the assumptions required for Thm.~\ref{main-theo3} to hold are verified by most examples of trace-invariants whose large $N$ asymptotics have been studied in the literature. They imply for instance the large $N$ factorization of planar trace-invariants for the complex Gaussian tensor with $D=3$ indices, a question raised in Ref.~\cite{Gurau:2025evo}. 

While the three main theorems -- Thm.~\ref{main-theo1}, Thm.~\ref{main-theo2} and Thm.~\ref{main-theo3} -- are conceptually simple, they provide considerable improvements on the aforementioned condition: they are of the same combinatorial nature, but have the significant advantage of only involving the polynomials $P_1, \ldots, P_s$ themselves.

\paragraph{Multipartite entanglement structure of random $D$-partite quantum states.} In quantum information, one may interpret a (normalized) complex tensor $S$ (as introduced above) as providing a coordinate-dependent definition of a \emph{$D$-partite quantum state} $\ket{S} \in (\mathbb{C}^{N})^{\otimes D}$. In this context, $\LU$ transformations parametrize independent local changes of frames in each of the $D$ subsystems constituting $\ket{S}$. The \emph{entanglement structure} of $\ket{S}$ can be defined as its orbit under the action of the group of $\LU$ transformations. Two states lying in the same $\LU$-orbit are then deemed $\LU$-equivalent, meaning that they share the same entanglement properties: indeed, this formalizes the intuitive idea that entanglement is precisely what classifies the non-local properties of $D$-partite quantum states. In the companion paper \cite{Carrozza:2026qcf}, we reviewed how trace-invariants are in general sufficient to separate $\LU$-orbits of deterministic states, and used that fact to extend the notion of $\LU$-equivalence to random quantum states $\ket{S}$ (which, after fixing local orthonormal bases, are represented by random tensors $S$): two random quantum states are $\LU$-equivalent if and only if their trace-invariants have the same correlation functions. The large $N$ factorization problem is thus clearly relevant to an \emph{asymptotic} version of $\LU$-entanglement classification of random states, in a large $N$ regime where one only retains information about the leading (resp.~leading and next-to-leading) order behavior of trace-invariants. In the bipartite setting ($D=2$), natural observables to consider in this context are the expectation values of entanglement R\'{e}nyi-$k$ entropies for integer $k$, which all take the form $- \ln P(S,\bar{S})$ where $P$ is a trace-invariant (that also happens to be positive). For a large class of random quantum states, such as the Haar random state defined by $T_U$ or generalizations known as \emph{random tensor networks} (which are built out from a collection of independent and identically distributed Haar random states), it can be shown that (see the original Ref.~\cite{Hayden:2016cfa}, and also Ref.~\cite{Cheng:2022ori} for generalizations): (i) at large $N$, the \emph{typical} value of $- \ln P(S,\bar{S})$ is $- \ln \mean{P(S,\bar{S})}$; (ii) moreover, the asymptotic relation
\begin{equation}\label{eq:average_renyi_intro}
    \mean{- \ln P(S,\bar{S})} \underset{N \to \infty}{\sim} - \ln \mean{P(S,\bar{S})}  
\end{equation}
holds \ie the asymptotic expectation of the entanglement R\'{e}nyi entropy associated to $P$ (which takes the form of an \emph{annealed} average) is directly determined by the asymptotic expectation of $P$ itself (or, in other words, by a \emph{quenched} computation that is much easier to perform than its annealed cousin). As was shown in Ref.~\cite{Carrozza:2026qcf}, in the multipartite setting ($D\geq 3$) it is possible to construct an infinite family of $\LU$-invariant functions which can be interpreted as multipartite generalizations of entanglement R\'{e}nyi entropies, in the sense that they take a similar functional form and can be proven to vary \emph{monotonically} under local operations. Any such quantity takes the form $- \ln \abs{P(S,\bar{S})}$, where $P$ is a trace-invariant that is in general complex-valued, and may in particular admit zeroes (in which case the value of the entropy at any zero is conventionally fixed to $+\infty$). 
When $D=2$, the validity of property (i) is intertwined with the large $N$ factorization property of bipartite trace-invariants, while property (ii) follows from additional uniform lower-bounds on those trace-invariants. Since neither of those results are universally true in the multipartite setting, extending properties (i) and (ii) to $D\geq 3$ requires extra care. Clarifying under which hypotheses such results can be proven to hold is the purpose of \textbf{Sec.~\ref{subsec:averageApprox}}. Our two main results in that section are: the concentration measure phenomenon of Prop.~\ref{prop:concentration}, which, for trace-invariants that obey a certain large $N$ factorization criterion (Eq.~\eqref{eq:Ansatz_large-N}), allows one to determine the typical value of the associated multipartite R\'{e}nyi entropies; and Prop.~\ref{prop:asymptotic_R_G}, which, under the additional assumption that a given trace-invariant obeys a power law (see Def.~\ref{def:power_law}), establishes a result analogous to Eq.~\eqref{eq:average_renyi_intro}. While our main focus in the present contribution is the application of those two results to Gaussian or Haar random tensors (see Cor.~\ref{cor:Lit_inv_and_exchange}), we expect them to be applicable to a large class of random tensors, including random tensor networks.

\paragraph{Acknowledgements.} We thank Razvan Gurau and Hannes Keppler for insightful discussions on the large $N$ factorization problem.
This work was supported by the  ANR JCJC project ``RTFPQuEnt" (ANR-25-CE40-5465). The IMB receives support from the EIPHI Graduate School (contract ANR-17-EURE-0002).  L.L. also acknowledges support  from the  ANR PRC project  ``TAGADA" (ANR-25-CE40-5672).

\section{Prerequisites and preliminary results} \label{sec:Prerequisites}

\paragraph{Colored graphs.} A bipartite  $D$-edge-colored graph $G$, or simply a \textit{$D$-colored graph}, is specified by a pair consisting of a vertex set and an edge set. The vertex set is bipartite, \ie it is the disjoint union of a set of white vertices and a set of black vertices, and edges can only link black and white vertices. The edge set is partitioned into $D$ disjoint subsets, each corresponding to a label in $\{1, \ldots, D\}$ called \emph{color}, and there is exactly one edge of each color incident to a given vertex.

As a consequence, a $D$-colored graph $G$ has the same number of black vertices and white vertices, which we denote $k(G)$. The set of $D$-colored graphs (resp.~connected $D$-colored graphs) is denoted by $\cG_D$ (resp.~$\cG_{D}^{\conn}$) and the number of connected components of a $D$-colored graph $G$ is denoted by $\kappa(G)$ (in particular, $\kappa(G)=1$ whenever $G \in \cG_{D}^{\conn}$). An example of a graph $G\in\cG_3^{\conn}$ is shown on the left of Fig.~\ref{fig:Graph-and-poly-ex}.

\begin{figure}[ht]
    \centering
    \includegraphics[height = 2cm]{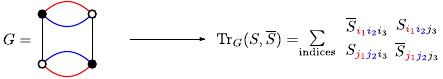}
    \caption{Left: a connected 3-colored graph $G$ in $\cG_3^{\conn}$. Right: the corresponding trace-invariant $\tr_G$. }
    \label{fig:Graph-and-poly-ex}
\end{figure}

The disjoint union of $G_1$ and $G_2$ -- denoted $G_1 \sqcup G_2$ -- is the $D$-colored graph whose connected components are exactly those of $G_1$ and those of $G_2$.

Given a multiset\footnote{That is, the elements of $\gG = \{G_1, \ldots, G_p\}$ are accounted for with multiplicities.} $\gG = \{G_1, \ldots, G_p\}$ where each $G_i \in \cG_D$, and letting $G=G_1\sqcup \cdots \sqcup G_p$, we define  $\cG_{D+1}(G)$ as the set of graphs $\widehat G\in\cG_{D+1}$ obtained by adding edges of color $0$ to G in all possible ways. Note that in this definition, the $G_i$ are not required to be connected (so $\gG$ is equivalent to the data of $G$ together with a partition of its connected components). We also define  $\cG_{D+1}^{\conn}(G)$ as the subset of connected graphs of $\cG_{D+1}(G)$. Moreover, given $\widehat G \in \cG_{D+1}(G)$, we can construct a graph $K(\widehat{G}; \gG)$ whose vertex-set is $\{G_i\}$, and whose edge-set is defined as follows: a pair $(G_i , G_j)$ (with $i\neq j$) is an edge of $K(\widehat{G}; \gG)$ if and only if there exists an edge of color $0$ in $\widehat{G}$ that connects a vertex of $G_i$ to a vertex of $G_j$. With this definition at hand, we will say that $\widehat G$ \emph{connects} $\gG = \{G_1, \ldots, G_p\}$ if and only if $K(\widehat{G}; \gG)$ is a connected graph. Note that if the $\{G_i\}$ are all themselves connected, then  $\cG_{D+1}^{\conn}(\gG)=\cG_{D+1}^{\conn}(G)$, but that this relation is not true in general. 

\ 

Given $G \in \cG_D$ and $1\le i<j\le D$, we denote by $F_{ij}(G)$ the number of connected components of the graph obtained by removing all the edges of $G$ whose colors differ from both $i$ and $j$. Because $G$ is $D$-edge-colored, this graph is $2$-edge-colored, and all its vertices have valency two: it is a collection of cycles called \emph{faces of color} $ij$. If $\widehat G\in\cG_{D+1}$ has edges with colors in $\{0, \ldots, D\}$,  the faces of color $0$ refer to  all of its faces of color $0i$ for $1\le i \le D$, and we let 
\begin{equation}
    F_0(\widehat G) = \sum_{i=1}^D F_{0i}(\widehat G),\hspace{1.2cm} \textrm{ and }  \hspace{1.2cm} F(G) = \sum_{1 \leq i<j \leq D} F_{ij}(G)\;. 
\end{equation} 
One also defines $F(\widehat G) = F(G) + F_0(\widehat G)$, as well as:
\begin{equation}
 \mF(G)= \max_{\widehat G \in \cG_{D+1}(G)}F_0(\widehat G),  \hspace{1.2cm} \textrm{ and }  \hspace{1.2cm}   \mFc(\gG)= \max_{\widehat G \in \cG^{\conn}_{D+1}(\gG)} F_0(\widehat G)  \;.
\end{equation}
The quantity $\mF(G)$ is related to the so-called \textit{degree of compatibility} studied in Refs.~\cite{Collins:2024pip,Carrozza:2026qcf} and defined as:
\begin{equation} \label{eq:Delta_0}
    \Delta(G)\ =  \  \frac{D(D-1)}4 k(G) + \frac 1 2 F(G) - \frac {D-1} 2 \mF(G)\,.  
\end{equation} 
For $\widehat G \in \cG_{D+1}(G)$, we denote by $\Delta_0(\widehat G)$ the quantity defined similarly as Eq.~\eqref{eq:Delta_0} replacing $\mF(G)$ by $F_0(\widehat G)$.

\

We finally introduce the following quantity for $G\in\cG_D$:
\begin{align}
    &\frac{2}{(D-2)!}\omega(G) = (D-1)\kappa(G) + \frac{(D-1)(D-2)}{2}k(G) - F(G)\,,
\end{align}
which is known in the literature on random tensors as the \emph{Gurau degree}.\footnote{In the reference textbook by Gurau \cite{gurau_random_2017}, $\omega$ is referred to as the \emph{degree}, and $\frac{2}{(D-2)!}\omega$ as the \emph{reduced degree}. As those two quantities merely differ by a normalization convention, either one of them is sometimes referred to as the \emph{Gurau degree} in publications by other authors. This is in particular the case in Ref.~\cite{Carrozza:2026qcf}, where $\frac{2}{(D-2)!}\omega$ was also called the $2$\textit{-complete degree} and denoted by $\omega_2$.} In Ref.~\cite[605]{GURAU2011592} and Ref.~\cite[46]{gurau_random_2017}, it was proven that for any $D$-colored graph $G$ and any $\widehat G \in \cG_{D+1}(G)$: $\omega(\widehat G) \ge D \omega(G)$. Interpreted in terms of faces of color $0$, this inequality can be reformulated as the following lemma.

\begin{lem}[Gurau \cite{GURAU2011592}] \label{lem:Razvans-bound-on-F}
    Let $D \geq 2$ and $G\in\cG_D$. For all $\widehat G \in \cG_{D+1}(G)$, we have:
    \begin{equation}
        \label{eq:Razvans-bound-on-F}
        F_0(\widehat G) \le \frac D 2 k(G) + \frac {F(G)} {D-1}  -D\bigl(\kappa(G) - \kappa(\widehat G)\bigr)\,, \qquad \rm{or} \qquad \Delta_0(\widehat G) \geq \frac{D(D-1)}{2}(\kappa(G) - \kappa(\widehat G)) \,.
    \end{equation}
\end{lem}

\paragraph{Trace-invariants and correlations.} Let $D \geq 2$. We consider a tensor  $S = \paa{S_{i_1 \dots i_D}}\in\mathbb{C}^{N^D}$ with $D$ ordered indices, such that the index $i_c$ in position $c$ -- said to have color $c$ -- ranges from $1$ to $N$, and its conjugate $\bar{S} = \paa{\bar{S}_{i_1 \dots i_D}}$. \textit{Trace-invariants} are a class of homogeneous polynomials in the tensor entries that can be put in one-to-one correspondence with $D$-colored graphs. 

To a given $D$-colored graph $G$, we associate the trace-invariant denoted by $\tr_G(S,\bar{S})$ as follows: consider a copy of $S$ for every white vertex of $G$ and a copy of $\bar S$ for every black vertex. For every edge of color $c$ of $G$, the indices of color $c$ of the corresponding copies of $\bar S$ and $S$ are then set as equal and summed from 1 to $N$. This results in a $\LU$-invariant homogeneous polynomial of degree $k(G)$ in $S$ (resp.~in $\bar S$). An example of graph $G$ together with its corresponding trace-invariant $\tr_G$ is shown in Fig.~\ref{fig:Graph-and-poly-ex}.

A $D$-colored graph $G$, or equivalently a trace-invariant $\tr_G$ can be (non-uniquely) encoded into a list of $D$ permutations $\vec \sigma = (\sigma_1,\ldots,\sigma_D)$ (see Ref.~\cite{BenGeloun:2013lim}). To define $\vec \sigma$, we start by assigning distinct labels from $1$ to $k(G)$ to all white (resp.~black) vertices of $G$. Then, for each label $s \in \paa{1,\dots,k(G)}$ and each color $c \in \paa{1,\dots,D}$, an edge of color $c$ links the white vertex with label $s$ to the black vertex with label $\sigma_c(s)$. This means that the graph $G$, together with its labeled vertices, can be fully encoded by the $D$-tuple of permutations $\vec \sigma \in S_{k(G)}^D$. So, there is a one-to-one correspondence between labeled $D$-colored graphs and $D$-tuple of permutations, that induces a one-to-one correspondence between unlabeled $D$-colored graphs and certain equivalence classes of $D$-tuples of permutations (for more information, see Ref.~\cite[11]{Carrozza:2026qcf} and references therein). We will denote by $S_{k(G)}^D(G)$ the set of $D$-tuples of permutations of $S_{k(G)}$ associated to $G$.

\

Let $T$ be a Gaussian random tensor whose components $T_{i_1 \dots i_D}$ are centered i.i.d. Gaussian complex random variables with variance $1/N^D$. The density of the distribution is then given by $\cZ^{-1} \e{-N^D \norm{T}^2}$, where $\displaystyle\norm{T}^2 =  \sum_{i_1 = 1}^N \cdots \sum_{i_D=1}^N T_{i_1 \ldots i_D} \bar{T}_{i_1 \ldots i_D}$ and, letting $\displaystyle\d T = \prod_{i_1=1}^N\cdots\prod_{i_D=1}^N \d T_{i_1 \dots i_D}$ and $\displaystyle\d \bar{T} = \prod_{i_1=1}^N\cdots\prod_{i_D=1}^N \d \bar{T}_{i_1 \dots i_D}$,
\begin{equation}
    \cal{Z} = \int \e{-N^D \norm{T}^2} \d T \d \bar{T}\,.
\end{equation}
For any function $f$ of the components of $T$  and its conjugate, the average -- or expectation --  of $f(T,\bar{T})$ is defined as:
\begin{equation}
    \mean{f(T,\bar{T})} = \int_{\bb{C}^{N^D}} f(T,\bar{T}) \, \e{-N^D \norm{T}^2}\frac{\d T \d \bar{T}}{\cZ} \,.
\end{equation}

Consider any deterministic tensor $T_0$ of norm  $\norm{T_0} = 1$, the random tensor $T_U = U T_0$, where $U$ is sampled according to the (normalized) Haar measure on the  unitary group $U(N^D)$,  will be  called \emph{Haar random tensor} (the distribution is independent from the particular choice of $T_0$). The average of a function $f(T_U, \bar T_U)$ is then expressed as:
\begin{equation}
    \mean{f(T_U,\bar{T}_U)} = \int_{U(N^D)} f(T_U, \bar{T}_U) \, \d U\,.
\end{equation}

For the complex Gaussian $T$, Wick's theorem applied to a trace-invariant translates into the following graph expansion:
\begin{equation}
\label{eq:exp-and-asympt-moments}
    \langle \tr_G(T,\bar{T}) \rangle = \sum_{\nu \in S_{k(G)}} N^{\sum_{c=1}^D\#(\sigma_c \nu^{-1}) - Dk(G)} = \sum_{\widehat G \in \cG_{D+1}(G)} \Xi(\widehat G) N^{F_0(\widehat G) - Dk(G)} \underset{N \to \infty}{\sim} \mu_G N^{\mF(G) - Dk(G)}\,,
\end{equation}
where, for a permutation $\pi$ of $S_k$, $\#(\pi)$ counts the number of cycles of $\pi$ and $\Xi(\widehat G)$ is a degeneracy factor defined, for $\vec\sigma \in S_{k(G)}^D(G)$, by:
\begin{equation}
    \Xi(\widehat G) = \abs{\paa{\nu \in S_{k(G)} \;\; \rm{s.t.} \;\;\pa{\sigma_1,\dots,\sigma_D,\nu} \in S_{k(G)}^{D+1}(\widehat G)}}  \,.
\end{equation}
Moreover, the factor $\mu_G$ counts the number of extremizers: fixing $\vec\sigma \in S_{k(G)}^D(G)$, we have
\begin{equation} \label{eq:mu_def}
    \mu_G = \abs{\paa{\nu \in S_{k(G)} \, \rm{ s.t. } \,\sum_{c = 1}^D \#(\sigma_c\nu^{-1}) = \mF(G) }}\,.
\end{equation}
Equivalently, one can introduce the set 
\begin{equation}
    \label{eq:MD}
    \MD{}(G) = \paa{\widehat G \in\cG_{D+1}(G) \;\; \rm{s.t.} \;\; F_{0}(\widehat G) = \mF(G)} \,,
\end{equation}
to formulate the combinatorial factor in terms of $\Xi$ as
\begin{equation} \label{eq:mu_vs_chi}
    \mu_G = \sum_{\widehat G \in \MD{}(G)} \Xi(\widehat G) \,,
\end{equation} 
We call \emph{pairing} of the vertices of a graph $G\in\cG_D$ a partition of the vertices of $G$ whose blocks each contain exactly one white vertex and one black vertex. 
For every $\widehat G^\textrm{opt}\in\MD{}(G) $ the edges of color 0 induce a pairing of $G$. We call the set of pairings induced by the graphs $\widehat G^\textrm{opt}\in\MD{}(G) $ the \emph{dominant pairings of $G$}. There is a one-to-one correspondence between this set and  $\MD{}(G)$. 

One can show (see \eg Ref.~\cite{Nechita:2007bpy,Carrozza:2026qcf}) that for any $D$-colored graph $G$ with $k=k(G)$: 
\begin{equation}
    \label{eq:Haar-vs-gaussian}
     \langle \tr_G(T_U,\bar{T}_U) \rangle = f_{k,D,N}\   \langle \tr_G(T,\bar{T}) \rangle\;, \qquad \mathrm{with} \qquad    f_{k,D,N}=\frac{N^{Dk}\pa{N^D - 1}!}{\pa{N^D - 1 + k}!}\ \underset{N \to \infty}{\longrightarrow}\   1\;,
\end{equation}
thus making the study of dominant pairings similar in both cases. In particular: 
\begin{equation}
\label{eq:equiv-Haar-gaussian}
    \langle \tr_G(T_U,\bar{T}_U) \rangle \underset{N \to \infty}{\sim}   \langle \tr_G(T,\bar{T}) \rangle \;.
\end{equation}
\begin{rem}
    In the literature on random tensor models, the variance of the tensor components of the Gaussian random tensor with $D$ indices is usually taken to be $1/N^{D-1}$. We choose the variance to be $1/N^D$ in order to have the asymptotic equivalence of Eq.~\eqref{eq:equiv-Haar-gaussian}.
\end{rem}

The connected correlations -- or cumulants -- of trace-invariants are defined for a distribution $\d\mu(S)$ as 
\begin{equation}
     \left\langle \tr_{G_1}(S,\bar{S}) , \ldots , \tr_{G_p}(S,\bar{S})  \right\rangle_{\rm{conn}} = \frac{\partial}{\partial t_1} \cdots \frac{\partial}{\partial t_p} \log \mean{ \rm{e}^{\sum_{i=1}^p t_i \tr_{G_i}(S,\bar{S}) }} \biggr\rvert_{t_1=\cdots =t_p=0}\;.
\end{equation}
For the complex Gaussian tensor $T$, they can be expressed as follows: for any $\gG = \{G_1 , \ldots, G_p\}$ 
\begin{equation} \label{eq:Exp_gG}
    \left\langle \tr_{G_1}(T,\bar{T}) , \ldots , \tr_{G_p}(T,\bar{T})  \right\rangle_{\rm{conn}} = \sum_{\widehat G \in \cG_{D+1}^{\conn}(\gG)} \Xi(\widehat G) N^{F_0(\widehat G) - Dk(\widehat G)}  \underset{N \to \infty}{\sim} \mu_{\gG}^{\conn} N^{\mFc(\gG) - Dk(\widehat G)} \,,
\end{equation}
where we note that $k(\widehat G) = \sum_{i=1}^p k(G_i)$. The combinatorial constant $\mu_\gG^{\conn}$ and the set $\MD{}^{\conn}(\gG)$ are defined in a similar fashion as $\mu_G$ and $\MD{}(G)$ by replacing $\mF(G)$ by $\mFc(\gG)$ and $\cG_{D+1}(G)$ by $\cG_{D+1}^{\conn}(\gG)$.

The general relations between the connected correlations of trace-invariants for $(T_U,\bar{T}_U)$ and those of $(T,\bar{T})$ are more involved than Eq.~\eqref{eq:equiv-Haar-gaussian}; exploring their implications is deferred to future work.

\paragraph{Large $N$ factorization.}
\begin{defi}
\label{def:facto}
We say that $\{G_1, \dots, G_p\}\in \cG_D$ \emph{factorizes at large} $N$ if for the complex Gaussian tensor $(T, \bar T)$:
\begin{equation}
\label{eq:large-N-facto}
   \langle \tr_G(T,\bar{T}) \rangle  \ \underset{N \to \infty}{\sim}\  \prod_{i=1}^p \ \langle \tr_{G_i} (T,\bar{T}) \rangle \,, 
\end{equation}
where $G = G_1 \sqcup \cdots \sqcup G_p$. If, in addition, the graphs $G_i$ are all connected, we say that $G$ \emph{factorizes over its connected components at large} $N$.    
\end{defi}

Note that, due to Eq.~\eqref{eq:equiv-Haar-gaussian}, $\{G_1, \dots, G_p\}\in \cG_D$ factorizes at large $N$ for the Haar random tensor if and only if it factorizes at large $N$ for the complex Gaussian tensor.

\

Fixing $G_1, \ldots, G_p\in\cG_D$ not necessarily connected and letting $G=G_1\sqcup \cdots \sqcup G_p $, one can write the moment-cumulant formula for the invariant $\tr_{G}$  as:
\begin{equation}
\label{eq:mom-cum-formula}
\begin{split}
    \langle \tr_G(T,\bar{T}) \rangle  &\ =\  \sum_{\pi\in\mathcal{P}(p)} \ \prod_{B\in\pi}  \    \left\langle \bigl\{\tr_{G_i}(T,\bar{T})\bigr\}_{i\in B}  \right\rangle_{\rm{conn}}\\&\ =\  \prod_{i=1}^p\  \langle \tr_{G_i}(T,\bar{T}) \rangle \  + \ \cdots\  +\      \left\langle \tr_{G_1}(T,\bar{T}) , \dots , \tr_{G_p}(T,\bar{T})  \right\rangle_{\rm{conn}}\;, 
\end{split}
\end{equation}
where, for $B = \paa{i_1,\dots,i_{\abs{B}}}$, we use the notation 
\begin{equation} \label{eq:partial_cum}
    \left\langle \bigl\{\tr_{G_i}(T,\bar{T})\bigr\}_{i\in B}  \right\rangle_{\rm{conn}} = \mean{\tr_{G_{i_1}}(T,\bar{T}),\dots,\tr_{G_{i_{\abs{B}}}}(T,\bar{T})}_{\rm{conn}} \underset{N \to \infty}{\sim} \mu_{\{G_i\}_{i \in B}}^{\conn} N^{\mFc(\{G_i\}_{i \in B}) - D\sum_{i \in B} k(G_i)}\,.
\end{equation}
We emphasize that for $\widehat G_B \in \cG_{D+1}^{\conn}(\{G_i\}_{i \in B})$, the label ``$\mathrm{conn}$'' does not refer to the connectedness of $\widehat G_B$ but instead to the connectedness of $K(\widehat G_B ; \{G_i\}_{i \in B})$. On the first line of Eq.~\eqref{eq:mom-cum-formula}, the product is then over the blocks $B$ of the partition $\pi$. On the second line of Eq.~\eqref{eq:mom-cum-formula}, the completely factorized term is obtained for $\pi$ the partition $0_p$ whose blocks each have one element, and the completely connected term for $\pi$ the one-block partition $1_p$. 
Remark that Eq.~\eqref{eq:facto-intro-conn-corr} is the particular case of Eq.~\eqref{eq:mom-cum-formula} for $p=2$. 

From this expansion, it is clear that \emph{the large $N$ factorization Eq.~\eqref{eq:large-N-facto} occurs if and only if the expansion Eq.~\eqref{eq:mom-cum-formula} is dominated by  the completely factorized term} (for $\pi=0_p$), while the other terms (for any other $\pi \neq 0_p$) in the expansion are suppressed in comparison, that is, Eq.~\eqref{eq:large-N-facto} holds if and only if
\begin{equation}
\label{eq:cond-on-faces-for-facto}
    \forall \pi \neq 0_p,\qquad \sum_{B\in \pi} \mFc\bigl(\{G_i\}_{i\in B}\bigr)  < \sum_{i=1}^p \mF(G_i)\;.
\end{equation}
In particular, if $\{G_1, \ldots, G_p\}$ factorizes at large $N$, then necessarily:
\begin{equation}
\label{eq:mF-equal-parts}
    \mF(G) = \sum_{i=1}^D \mF(G_i)\;,
\end{equation}
but this condition is \apriori unlikely to be sufficient, as for some $\pi \neq 0_p$, one might still have $\mF(G)  = \sum_{B\in \pi} \mFc\bigl(\{G_i\}_{i\in B}\bigr)$. And, indeed, this intuition will be confirmed below.

\paragraph{Boundary graphs and factorization conditions from Ref.~\cite{Gurau:2025evo}.}
If $G\in \cG_D$, we let $\cG_{D+1}^\circ(G)$ be the set of graphs    obtained from $G$ by adding a certain number $k'< k(G)$ of edges of color 0 linking some black and white vertices of $G$. Equivalently, any graph in $\cG_{D+1}^\circ(G)$ can be obtained by removing $k(G) - k'$ color-$0$ edges from a graph $\widehat G\in \cG_{D+1}(G)$. If $\mathring{G}\in \cG_{D+1}^\circ(G)$, we call \emph{boundary vertices} the  subset of its vertices with no edge of color 0 attached. For every $c\in\{1, \ldots, D\}$, after removing all the edges that are not of color 0 or $c$ from $\mathring{G}\in \cG_{D+1}^\circ(G)$, all the vertices have valency one or two: any connected component is either a cycle -- called an \emph{internal face} -- or an open path connecting a pair of boundary vertices of $\mathring{G}$ -- called an \emph{external face}. We will denote by $F_0 (\mathring{G})$ the number of internal faces of $\mathring{G}$ (which consistently extends the previous definition of $F_0$).

The boundary graph $\partial \mathring{G}$ of $\mathring{G}$ is built as follows: its vertices are the boundary vertices of $\mathring{G}$, and for every external face of color $0c$ of $\mathring{G}$ whose extremities are the boundary vertices $v, v'$, the boundary graph $\partial \mathring{G}$  has an edge of color $c$ linking $v$ and $v'$. Note that if $\mathring{G}=G$ ($k'=0$), then $\partial  \mathring{G} = G$. An example of boundary graph is given in Fig.~\ref{fig:boundaryGraph}. We can finally introduce an auxiliary graph $K(\mathring{G};\gG)$, defined in a similar way as $K(\widehat{G}; \gG)$ above, such that $\mathring{G}$ will be said to connect $\gG$ whenever $K(\mathring{G};\gG)$ is connected. More broadly, the integer-valued map $|K(\,\cdot\, ; \gG)|$ can now be applied to graphs with empty and non-empty boundaries alike, and provides a quantitative measure of how well such graphs connect $\gG$. 

\begin{figure}[ht]
	\centering
	\includegraphics[width = \textwidth]{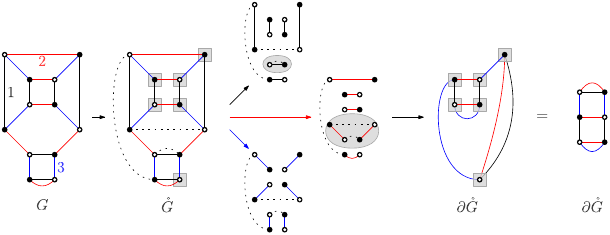}
	\caption{From left to right: a $3$-colored graph $G\in \cG_3$. A graph $\mathring{G} \in \cG_{4}^\circ(G)$ with $k' = 3 < 6 = k(G)$ and where the boundary vertices have been highlighted by gray boxes. The three restrictions of $G$ to colors $01$, $02$, and $03$, and where the internal faces have been stressed by the use of gray bubbles. The boundary graph $\partial \mathring{G}$ where we emphasize the boundary vertices. A simplified representation of the boundary graph $\partial\mathring{G}$.}
	\label{fig:boundaryGraph}
\end{figure}

In the paragraph ``A necessary and sufficient condition for factorization.'' of Ref.~\cite{Gurau:2025evo}, conditions can be extracted to imply factorization or non-factorization for a real random tensor. These conditions can be easily adapted to the complex case, as presented in the following lemma.

\begin{lem} \label{lem:conditions_Gurau}
    Let $D \geq 2$ and $H,G_1,\ldots,G_p \in \cG_D$ be $D$-colored graphs. We let $\bar H$ be the $D$-colored graph obtained from $H$ by flipping the colors (black or white) of all its vertices.\footnote{While for $D=2$, $H=\bar H$, this is, \lat{a priori}, not the case for $D>2$, and an explicit counterexample is given for $D=6$ in the last section of Ref.~\cite{Lionni2018}. However, one always has $\tr_{\bar{H}}(\cdot) = \overline{\tr_H (\cdot)}$, \ie $\mF(H) = \mF(\bar H)$.}
    \begin{enumerate}
        \item If $\mF(H) \le \frac D 2 k(H)$, then, $\{H,\bar H \}$ does not factorize at large $N$, \ie 
        \begin{equation}\label{eq:no-facto-Razvan}
            \bigl\langle \tr_H(T, \bar T) \tr_{\bar H}(T, \bar T) \bigr\rangle  \underset{N \to \infty}{\sim}    \bigl\langle \tr_H(T, \bar T) \tr_{\bar H}(T, \bar T) \bigr\rangle_{\mathrm{conn}} \,.
        \end{equation}
        \item If, for any $i \in \paa{1,\dots,p}$ and any boundary graph $\mathring{G_i} \in \cG_{D+1}^\circ(\mathring{G_i})$, we have
        \begin{equation}\label{eq:Razvans-condition-for-facto}
            \mF(\partial \mathring G_i ) > \frac D 2 k(\partial \mathring G_i) \,,
        \end{equation}
        then $\{G_1, \ldots, G_p\}$ factorizes at large $N$.
    \end{enumerate}
\end{lem}

\begin{proof}
    The proof is inspired by Ref.~\cite{Gurau:2025evo}. Regarding point 1, we know that: if $\mathbf{H} = \{H, \bar H\}$ factorizes at large $N$, then, $\mF(H) + \mF(\bar H) > \mFc(\mathbf{H})$. Among the pairings that connect $H$ and $\bar H$, the one that connects any vertex $v$ of $H$ to its image $\bar v$ in $\bar H$ (see the right-hand side of Fig.~\ref{fig:MST_incomp}), and denoted by $\widehat{H\sqcup \bar H}$, satisfies:
    \begin{equation}
        F_0(\widehat{H\sqcup \bar H}) = Dk(H) \,.
    \end{equation}
    Therefore, $\mFc(\mathbf{H}) \geq Dk(H)$. Since $\mF(H) = \mF(\bar H)$ we have that: if $\mathbf{H} = \{H, \bar H\}$ factorizes, then, $2 \mF(H) > Dk(H)$. The contrapositive of the latter proves the first statement of Lem.~\ref{lem:conditions_Gurau}.
    
    We now address the second statement of Lem.~\ref{lem:conditions_Gurau}. Assume that for any $i \in \paa{1,\dots,p}$, and any boundary graph $\mathring{G_i} \in \cG_{D+1}^\circ(\mathring{G_i})$, the following holds: 
    \begin{equation} \label{eq:assumpt_proof}
        \mF(\partial \mathring G_i ) > \frac D 2 k(\partial \mathring G_i) \,.
    \end{equation}
    To establish the large $N$ factorization, it suffices to verify that Eq.~\eqref{eq:cond-on-faces-for-facto} is satisfied. \\
    Let $G = G_1 \sqcup \cdots \sqcup G_p$, $\gG = \{G_1,\ldots,G_p\}$ and consider a non-trivial partition $\pi \neq 0_p$ of $\{1, \ldots, p\}$. For any $B \in \pi$, we can find $\widehat{G}_B \in \cG_{D+1}^{\conn}(\{G_i\}_{i \in B})$ such that $F_0(\widehat G_B) = \mFc(\{G_i\}_{i \in B})$ (see Eq.~\eqref{eq:partial_cum}). We then define $\displaystyle\widehat G = \bigsqcup_{B \in \pi} \widehat G_B \in \cG_{D+1}(\gG)$. \\
    For each $i\in \paa{1,\ldots,p}$, let $\mathring{G_i} \in \cG_{D+1}^\circ(G_i)$ be the graph obtained from $\widehat G$ by retaining the graph $G_i$ and all the color-$0$ edges of $\widehat{G}$ whose endpoints are both vertices in $G_i$. Let us fix a block $B \in \pi$. If $B$ is a trivial block (\ie $B=\{i_0\}$ for some $i_0\in \{1, \ldots, p\}$), one has the obvious bound
    \begin{equation}\label{eq:trivial_bound}
       F_0 (\widehat{G}_B) \leq  \mF(G_{i_0}) = \sum_{i \in B} \mF(G_i)\,.
    \end{equation}
    Assume now that $B$ is a non-trivial block, \ie $|B|\geq 2$. One can decompose the set of faces of $\widehat{G}_B$ as
    \begin{equation}
        F_0(\widehat{G}_B) = \sum_{i \in B} F_0(\mathring{G}_i) +  F_0\pa{\widehat{\bigsqcup_{i \in B} \partial \mathring{G_i} }}\,,
    \end{equation}
    where $\displaystyle\widehat{\bigsqcup_{i \in B} \partial \mathring{G_i} }$ is an element of $\cG_{D+1}^{\conn}(\{\mathring{\partial G_i}\}_{i \in B})$. In this equation, the term $F_0(\mathring{G}_i)$ accounts for all the faces that are internal to some $\mathring{G}_i$, while the last term counts all the faces that intersect at least two elements of $\{G_i\}_{i \in B}$. By construction, for any $c \in \{1, \ldots, D\}$, any face of color $0c$ in $\displaystyle\widehat{\bigsqcup_{i \in B} \partial \mathring{G_i}}$ must contain at least two edges of color $0$ (otherwise this face would be internal to some $\mathring{G}_i$). But the total number of color-$0$ edges in $\displaystyle\widehat{\bigsqcup_{i \in B} \partial \mathring{G_i}}$ is $\displaystyle\sum_{i \in B} k(\mathring{\partial G_i})$; as a result, one necessarily has
    \begin{equation}
         F_0\pa{\widehat{\bigsqcup_{i \in B} \partial \mathring{G_i} }} \leq \frac{D}{2} \sum_{i \in B}  k(\mathring{\partial G_i})< \sum_{i \in B} \mF(\mathring{\partial G_i}) \,,
        \end{equation}
    where we have used Eq.~\eqref{eq:assumpt_proof} to obtain the second inequality. Thus,
    \begin{equation}
        F_0(\widehat{G}_B) < \sum_{i \in B} \left( F_0(\mathring{G_i}) + \mF (\mathring{\partial G_i})\right).
    \end{equation}
    Given that, for any $i \in B$, we can  clearly find a graph $\widehat{G_i}\in \cG_{D+1}(G_i)$ such that $F_0(\widehat{G_i})= F_0(\mathring{G_i}) + \mF (\mathring{\partial G_i})$, it follows that
    \begin{equation}\label{eq:less_trivial_bound}
        F_0(\widehat{G}_B) < \sum_{i \in B} \mF(G_i)\,.
    \end{equation}
    
    Finally, combining the bounds from Eqs.~\eqref{eq:trivial_bound} and \eqref{eq:less_trivial_bound}, and noting that $\pi$ contains at least one non-trivial block, we obtain:
    \begin{equation}
         \sum_{B\in \pi} \mFc (\{G_i\})_{i \in B} = F_0(\widehat G)  = \sum_{B \in \pi} F_0(\widehat{G}_B) < \sum_{B \in \pi} \sum_{i \in B} \mF(G_i)  = \sum_{i=1}^p \mF(G_i) \,.
    \end{equation}
    This proves that Eq.~\eqref{eq:cond-on-faces-for-facto} holds for any non-trivial partition $\pi$, and therefore, that $\{G_1, \ldots, G_p\}$ factorizes at large $N$. 
\end{proof}

\paragraph{Tree-like families and the two-cut property.}
 
The following concepts introduced in Refs.~\cite{Lionni2018, Bonzom:2018btd} will be important for the second and third main theorems. 
 
\begin{defi} \label{def:DEF}
Consider $\widehat G\in\cG_{D+1}^{\conn}(\gG)$,  $\gG=\{G_i\}_{i=1, \ldots, p}$, $G_i\in\cG_D$. We define the following notions:
\begin{enumerate}
    \item A two-cut in $\widehat G$ is a pair of edges whose removal increases the number of connected components of $K(\widehat{G}; \gG)$. In particular, when the graphs $G_i$ are all connected, a \emph{two-cut} in $\widehat{G}$ is more simply a pair of edges whose removal increases the number of connected components of $\widehat{G}$.
    
     \item $G_i$ satisfies the  maximal two-cut property in $\widehat G$ if there exists a dominant pairing $\pi_\textrm{opt}(G_i)$ such that for each pair of vertices in $\pi_\textrm{opt}(G_i)$, either the edge of color 0 links the two vertices in the pair, or the two edges of color 0 attached to these vertices in $\widehat G$ form a two-cut. 
     
     \item We say that $\widehat G$ is tree-like on $\gG$ if $\widehat G\in\cG_{D+1}^{\conn}(\gG)$ and every $G_i$ satisfies the maximal two-cut property in $\widehat G$.
     
     \item\label{it:HAVEtree} $\gG$ is said to have tree-like dominant pairings if the graphs in $\MD{}^{\conn}(\gG)$ (that is, the graphs that dominate the expansion Eq.~\eqref{eq:Exp_gG} of $\langle \tr_{G_1}(T, \bar T), \ldots,  \tr_{G_p}(T, \bar T) \rangle_{\mathrm{conn}}$) contain the graphs $\widehat G$ that are tree-like  on $\gG$.\footnote{Importantly, if the graphs $G_i$ are not connected,  this does not imply that $\MD{}^{\conn}(\gG)$ contains the $\widehat G$ that are tree-like on the connected components of $G_i$. }
    
    \item $\gG$ is said to have only tree-like dominant pairings if $\MD{}^{\conn}(\gG)$ is exactly the set of $\widehat G$ that are tree-like  on $\gG$.
\end{enumerate}
\end{defi}

Given two edges $e,e'$ of color 0 in a graph $\widehat G$, we call \emph{flip of $e$ and $e'$} the graph operation which consists in exchanging these two edges of color 0 in the only possible way which preserves bipartiteness. 

Equivalently, any $\widehat G$ that is tree-like on $\gG$ can be obtained from a collection of dominant graphs $\widehat{G_i}^\textrm{opt} \in \MD{}(G_i)$ ($1\leq i \leq p$), that are subsequently connected by a tree-like structure of flips. In more detail, choosing \eg $H_1\eqdef\widehat{G_1}^\textrm{opt}$ as the root of the tree-like structure, one can recursively build it out as follows: 1) we initialize the tree-like structure as ${\widehat G}^{(1)}\eqdef H_1$; 2) recursively, for any $k \in \{2, \ldots , p\}$ one selects an element $H_k$ of $\{\widehat{G_i}^\textrm{opt}\}_{1\leq i \leq p}\setminus \{H_1, \ldots , H_{k-1}\}$, an edge $e_{k}$ of ${\widehat G}^{(k-1)}$ as well as an edge $\tilde{e}_k$ of $H_k$, and define $\widehat{G}^{(k)}$ as the graph obtained by flipping the pair $(e_{k},\tilde{e}_{k})$ in $H_k \sqcup \widehat G^{(k-1)}$; 3) finally, we set $\widehat{G}\eqdef \widehat{G}^{(p)}$. Any tree-like $\widehat{G}$ on $\gG$ can be obtained in this way for some collection $\{\widehat{G_i}^\textrm{opt}\}_{1\leq i \leq p}$, and this is the reason why we use the qualifier \emph{tree-like}.

\begin{lem}
 \label{lem:faces-when-flip}
 Consider $\widehat G_1, \widehat G_2\in \cG_{D+1}$ and choose a color-$0$ edge in each graph: $e_1$ in $\widehat G_1$ and $e_2$ in  $\widehat G_2$. Flipping $e_1$ and $e_2$ in $\widehat{G}_1 \sqcup \widehat{G}_2$, one obtains a graph $\widehat G$ that satisfies:
 \begin{equation}
     F_0(\widehat G) = F_0(\widehat G_1) + F_0(\widehat G_2) - D\;. 
 \end{equation}
\end{lem}
\proof We let $e, e'$ be the two edges of color 0 of $\widehat G$ resulting from the flip of $e_1, e_2$. Flipping $e, e'$ in $\widehat G$, one recovers $e_1, e_2$ and  $\widehat G_1 \sqcup\widehat G_2$. The only faces  affected in this last operation are the ones that contain the edges $e$ or $e'$. Since $e$ and $e'$ form a two-cut in $\widehat G$, they belong to the same face of colors $0i$ of $\widehat G$ for every $1\le i \le D$, while after the flip, since $e_1$ and $e_2$ are in two different connected components, they each belong to $D$ face of colors $0i$. Therefore,  $F_0(\widehat G) = F_0(\widehat G_1) + F_0(\widehat G_2) - D$.  \qed

\ 
 
\begin{lem}[Lionni \cite{Lionni2018}, Sec.~4.2.6]
\label{lem:faces-of-treelike}
 Consider  $\gG=\{G_i\}_{i=1, \ldots, p}$ with $G_i\in\cG_D$, and $\widehat G\in\cG_{D+1}^{\conn}(\gG)$.  If $\widehat G$ is tree-like on $\gG$, then:
 \begin{equation}
 \label{eq:faces-of-tree-like}
     F_0(\widehat G) = D + \sum_{i=1}^p \bigl( \mF(G_i) - D\bigr)\;.
 \end{equation}
 In particular, all such $\widehat G$ have the same $F_0$. 
\end{lem}
\proof This is easily seen by induction: for $p=1$, $\widehat G=\widehat{G_1}$ so that  $F_0(\widehat G)= F_0(\widehat {G_1})$. Since $G_1$ satisfies the maximal two-cut property in  $\widehat{G_1}$:    $\widehat{G_1}\in \MD{}(G_1)$, so that   $F_0(\widehat {G_1})= \mF(G_1)$, so that Eq.~\eqref{eq:faces-of-tree-like} holds. Assume now that $p\geq 2$, that Eq.~\eqref{eq:faces-of-tree-like} holds for any $p'<p$, and consider $\widehat G\in\cG_{D+1}^{\conn}(\gG)$ tree-like on $\gG$. There exists a pair of edges of color 0 in $\widehat{G}$ attached to the vertices of a dominant pairing and which form a two-cut. Flipping these two edges, one obtains two smaller graphs which are both tree-like  on the $G_i$ that they contain. We let $I_1, I_2$ be the corresponding subsets of $\{1, \ldots, p\}$.  In this process, from Lem.~\ref{lem:faces-when-flip}, $F_0$ increases by  exactly $D$. By induction:
\begin{equation}
    F_0(\widehat G) = \Bigl[D + \sum_{i\in I_1} \bigl( \mF(G_i) - D\bigr) \Bigr] +   \Bigl[D + \sum_{i\in I_2} \bigl( \mF(G_i) - D\bigr) \Bigr] - D\;,
\end{equation}
which gives the desired result.   \qed

\

\begin{lem}[Lionni \cite{Lionni2018}, Sec.~4.2.6]
\label{lem:tree-dom-iff-faces}
With the notations above, $\gG$ has tree-like dominant pairings if and only if
 \begin{equation}
     \mFc(\gG) = D + \sum_{i=1}^p \bigl( \mF(G_i) - D\bigr). 
 \end{equation}
\end{lem}
\proof All the graphs  $\widehat G \in \MD{}^{\conn}(\gG)$  have $F_0(\widehat G) =  \mFc(\gG)$. If $\gG$ has tree-like dominant pairings, then $\MD{}^{\conn}(\gG)$ contains the $\widehat G$ that are tree-like  on $\gG$, which satisfy Eq.~\eqref{eq:faces-of-tree-like}. Reciprocally, if $  \mFc(\gG) = D + \sum_{i=1}^p \bigl( \mF(G_i) - D\bigr)$, then all the $\widehat G$ that are tree-like  on $\gG$ have $F_0(\widehat G) =  \mFc(\gG)$, so they all belong to $\MD{}^{\conn}(\gG)$. \qed 

\

We will need the lemmas regarding tree-like graphs. 

\begin{lem}
\label{lem:tree-like-subsets}
With the notations above, if $\gG$ has tree-like dominant pairings (resp.~only tree-like dominant pairings), then so do its subsets $\{G_i\}_{i\in I}$, where $I\subset \{1, \ldots, p\}$. 
\end{lem}
\proof Let  $I\subset \{1, \ldots, p\}$, $\displaystyle G_I=\bigsqcup_{i\in I} G_i$, and consider a graph $\widehat G_I\in \cG_{D+1}^{\conn}(\gG_I)$. We also let $I'=\{1, \ldots, p\} \setminus I$,  $\displaystyle G_{I'}=\bigsqcup_{i\in I'} G_i$, $\gG_{I'}=\{G_i\}_{i\in I'}$, and consider $\widehat G_{I'}\in \cG_{D+1}^{\conn}(\gG_{I'})$ that is tree-like on $\gG_{I'}$. 

Choosing any pair of edges  $e$ in $\widehat G_I$ and $e'$ in $\widehat G_{I'}$ of color 0 and flipping them, we obtain a graph $\widehat G\in \cG_{D+1}^{\conn}(\gG)$ which from Lem.~\ref{lem:faces-when-flip} has $F_0(\widehat G) = F_0(\widehat G_I) + F_0(\widehat G_{I'}) - D$ faces of color 0. Assuming $\gG$ to have tree-like dominant pairings (resp.~only tree-like dominant pairings), from Lem.~\ref{lem:tree-dom-iff-faces}, 
 $F_0(\widehat G)\le D + \sum_{i=1}^p \bigl( \mF(G_i) - D\bigr)$ with equality for tree-like dominant pairings (resp.~only for tree-like dominant pairings). On the other hand, since $\widehat G_{I'}$ is a tree-like graph on $\gG_{I'}$ by assumption, by  Lem.~\ref{lem:faces-of-treelike} we must have $F_0(\widehat G_{I'})= D + \sum_{i\in I'} \bigl( \mF(G_i) - D\bigr)$. Combining these equations, we obtain: 
\begin{equation}
    F_0(\widehat G_I) = D+ F_0(\widehat G) - F_0(\widehat G_{I'})  \le D + \sum_{i\in I} \bigl( \mF(G_i) - D\bigr)\;.
\end{equation}
Since this bound holds for any $\widehat G_I\in \cG_{D+1}^{\conn}(\gG_I)$, and is saturated by tree-like graphs (resp.~only by tree-like graphs), we conclude that $\displaystyle\mFc(\gG_I) = D + \sum_{i\in I} \bigl( \mF(G_i) - D\bigr)$. Hence, by  Lem.~\ref{lem:tree-dom-iff-faces}, $G_I$ has tree-like dominant pairings (resp.~only tree-like dominant pairings).  \qed

\ 

Let $G_1,\dots,G_p \in \cG_D$, we denote by $G = G_1 \sqcup \cdots \sqcup G_p$ and $\gG = \paa{G_1,\dots,G_p}$. For any $\widehat G \in \cG_{D+1}(G)$, we define the bipartite graph $\Gamma(\widehat G; \gG)$ as follows: it has white vertices for the connected components of $\widehat G$ and black vertices for each graph $G_i \in \gG$. For every $i \in \paa{1,\dots,p}$ and every connected component $B$ of $G_i$, observe that $B$ is contained in a unique connected component of $\widehat G$. For every such $B$, we draw  in $\Gamma(\widehat G; \gG)$ an edge between the black vertex $v_i$ associated with the graph $G_i$, and the white vertex $w$ associated with the connected component of $\widehat G$ containing $B$. Consequently, $\widehat G$ connects $\gG$, if and only if the bipartite graph $\Gamma(\widehat G; \gG)$ is connected.

\begin{lem}
\label{lem:max-of-kappa-0}
 Consider $\gG=\{G_i\}_{i=1, \ldots, p}$ with $G_1,\dots,G_p\in\cG_D$, and $\widehat G\in\cG_{D+1}^{\conn}(\gG)$. Then the following upper bound holds: 
 \begin{equation}
         \label{eq:upper-bound-kappa}
     \kappa(\widehat G) \le \kappa(G) - p + 1\;,
 \end{equation}
with equality if and only if $\Gamma(\widehat G; \gG)$ is a tree. 
\end{lem}
\proof Any connected graph with $E$ edges and $V$ vertices satisfies $E-V+1\ge 0$, with equality if and only if it has no cycles, \ie if the graph is a tree. For $\gG=\{G_i\}_{i=1, \ldots, p}$ with $G_i\in\cG_D$, and $\widehat G\in\cG_{D+1}^{\conn}(\gG)$, the number of edges and vertices of $\Gamma(\widehat G; \gG)$ are respectively $\kappa(G)$ and $p + \kappa(\widehat G)$, which proves the lemma.
\qed

\ 

For $i \in \paa{1,\dots,p}$, we let  $\{B_i^{(j)}\}_{1 \leq j \leq \kappa(G_i)}$ be the connected components of a graph $G_i$, so that $G_i = \bigsqcup_j B_i^{(j)} $,   $1\le j \le \kappa(G_i)$.
Furthermore, for any $i \in \paa{1,\dots,p}$, we define $\gG'(i)$ as: 
\begin{equation}
\label{eq:expl-of-Gamma-tree}
    \gG'(i) = ( \gG\setminus \{G_i\})\cup\{B_i^{(1)}, \ldots, B_i^{(\kappa(G_i))}\}\,.
\end{equation}
Another equivalent formulation of the statement ``$\Gamma(\widehat G; \gG)$ is a tree'' is that ``the graph $K(\widehat G; \gG'(i))$ has $\kappa(G_i)$ connected components'' (see Fig.~\ref{fig:Gprime}). Said otherwise, for every $i\in \paa{1,\dots,p}$, the graphs $B_i^{(j)}$ for $1\le j \le \kappa(G_i)$ belong to different connected components of the graph $\widehat G$, and the graph $G_i$ is the only ``connection’’ between these connected components.

\begin{figure}[htbp]
    \centering
    \includegraphics[height = 5.5cm]{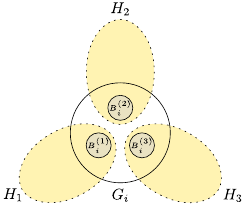}
    \caption{In such a graph $\widehat G$, every graph $G_i$ is the only connection point between the connected components $H_r$, for $1 \leq r \leq \kappa(G_i)$, of $\widehat G$ relative to $\gG'(i)$ (see Eq.~\eqref{eq:expl-of-Gamma-tree}).}
    \label{fig:Gprime}
\end{figure}

The following result and its proof are very similar to Prop.~4.2.8 of Ref.~\cite{Lionni2018}.

\begin{lem}
\label{lem:tree_like_B_implies_G}
    Let $D \geq 2$, $p \in \bb{N}^*$, and $G_1,\dots,G_p \in \cG_D$. For any $i\in \paa{1,\dots,p}$, let us denote by $\{B_i^{(j)}\}_{1 \leq j \leq \kappa(G_i)}$ the connected components of $G_i$, so that $G_i = \bigsqcup_j B_i^{(j)} $, where $1\le j \le \kappa(G_i)$.
   We assume that $\mathbf{B} = \{B_i^{(j)}\ \mid \ 1\le i \le p,\  1\le j \le \kappa(G_i)\}$ has tree-like dominant pairings. Then a graph $\widehat G\in \cG_{D+1}^{\conn}(\gG)$ belongs to $\MD{}^{\conn}(\gG)$ if and only if the two following conditions are satisfied:
\begin{enumerate}
\item Each connected component $C$ of $\widehat G$ belongs to $\MD{}^{\conn}(\mathbf{B}_C )$, where $\mathbf{B}_C$ is the subset of $\mathbf{B}$ on which  $C$ is supported. 
\item $\Gamma(\widehat G; \gG)$ is a tree. 
\end{enumerate}
Applied to each $G_i$ separately, this implies that the  dominant pairings of the $G_i$ are those induced by the $ \{B_i^{(j)}\}_{1\le j \le \kappa(G_i)}$ (that is, $G_i$ has no dominant pair with vertices on $B_i^{(j_1)}$ and $B_i^{(j_2)}$, $j_1\neq j_2$).

Furthermore, due to condition 2, if each $C$ is tree-like on $\mathbf{B}_C$, then $\widehat G$  is tree-like on  $\gG$. Therefore, if $\mathbf{B}$ has (only) tree-like dominant pairings, then so does $\gG$. 
\end{lem}

\proof  Fix $\kappa\ge 1$ and a graph $\widehat G\in \cG_{D+1}^{\conn}(\gG)$ with $\kappa(\widehat G)=\kappa$. One has 
\begin{equation}
    F_0(\widehat G) = \sum_{C}F_0(C)\,,
\end{equation}
where the sum is over the connected components $C$ of $\widehat G$. Since $\mathbf{B}$ has tree-like dominant pairings, so do its subsets by Lem.~\ref{lem:tree-like-subsets}, and therefore 
\begin{equation}
    F_0(\widehat G) \le  D\kappa + \sum_{i, j}\bigl(\mF(B_i^{(j)}) - D\bigr) \,.
\end{equation}
with equality if and only if the condition {\it 1} of the statement is satisfied. By Lem.~\ref{lem:max-of-kappa-0}, this is itself bounded as:
\begin{equation}
    F_0(\widehat G) \le D\Bigl( \sum_{i=1}^p\kappa(G_i) - p + 1\Bigr) + \sum_{i, j}\bigl(\mF(B_i^{(j)}) - D\bigr) =  \sum_{i, j}\mF(B_i^{(j)}) - D(p-1)  \,,
\end{equation}
with equality if and only if both conditions {\it 1} and {\it 2} of the statement are satisfied.
\qed

\ 

\paragraph{Some examples of colored-graphs.} Below, we define several families of colored graphs, some of which our main theorems will apply to. Examples are illustrated in Figs.~\ref{fig:melo} and \ref{fig:Graphs-ex}.

\begin{itemize}
    \item \textbf{ Melonic graphs:} Starting from the $2$-vertex graph (see leftmost graph of Fig.~\ref{fig:melo}), \textit{a connected melonic graph} is constructed recursively from insertions (or flips) of $2$-vertex graphs. An example is provided in Fig.~\ref{fig:melo}.
    \begin{figure}[htbp]
        \centering
        \includegraphics[height = 2.5cm]{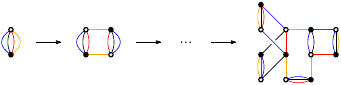}
        \caption{Recursive construction of a connected melonic graph.}
        \label{fig:melo}
    \end{figure}
    
    \item {\bf Maximally single-trace graphs:} They are the $D$-colored graphs that have a single face of colors $i,j$ for every $i,j\in\{1, \ldots, D\}$. 
    
    \item {\bf Cyclic graphs:} Choosing a subset $M$ of $\{1, \ldots, D\}$ with $\lfloor D/2\rfloor$ elements or less, the cyclic graphs associated to $M$ are cycles alternating $|M|$ parallel edges with colors in $M$, and $D-|M|$ parallel edges with colors in $\{1, \ldots, D\}\setminus M$. 
    
    \item {\bf Planar graphs:} We only define them for $D=3$. Connected planar graphs can be drawn on the plane without crossings, so that the clockwise (resp.~counterclockwise) ordering of the edges around the black (resp.~the white) vertices is $1,2,3$, cyclically. Equivalently, planar graphs satisfy\footnote{Indeed, $G$ being planar is characterized by the Euler equation $F(G)-E(G)+V(G)=2 \kappa(G)$, where $E$ (resp.~$V$) denotes the number of edges (resp.~vertices). But $V(G)= 2 k(G)$, and $G$ being $3$-regular, we also have $3 V(G)= 2 E(G)$. Eq.~\eqref{eq:planar} follows.}
   \begin{equation}\label{eq:planar}
       F(G)= 2\kappa(G) + k(G)\;.
   \end{equation}
    \item {\bf Compatible graphs:} A graph $G\in\cG_D$ is called \emph{compatible} if and only if $\Delta(G)=0$, where $\Delta$ has been defined in Eq.~\eqref{eq:Delta_0}.  Equivalently, $G$ is compatible whenever $\mF(G) = \frac{D}2 k + \frac 1 {D-1} F(G)$
    
    \item {\bf Realignment moments:} They are also called \emph{bi-pyramids} in Ref.~\cite{Lionni2018}. Choose a partition of $\{1, \ldots, D\}$ in three non-empty subsets $M_1, M_2$, and $M_3$. Take $k$ pairs of (black and white) vertices linked by $|M_3|$ edges with colors in $M_3$, then form a cycle from these pairs that are alternatively linked by $2|M_1|$ edges with colors in $M_1$ and by $2|M_2|$ edges with colors in $M_2$. 
    
    \item {\bf Joint realignment moments:} They are very similar to the realignment moments (they have a similar cyclic structure), but without the constraint that two subset of colors $M_1, M_2$  alternate along the cycle. 
\end{itemize}

\begin{figure}[htbp]
    \centering
    \includegraphics[width = .8\textwidth]{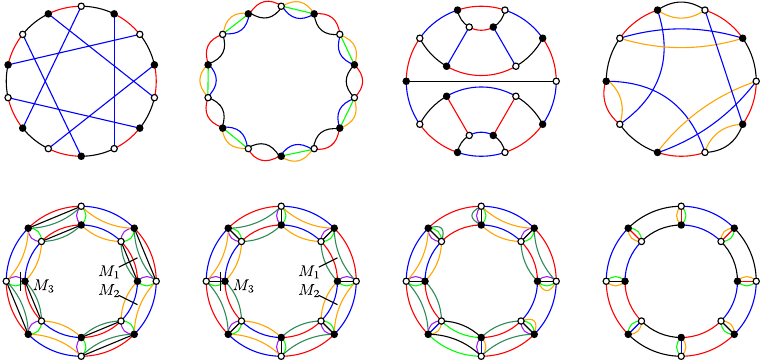}
    \caption{From left to right and top to bottom: a (compatible) maximally single-trace graph, a cyclic graph, a planar graph (with $\Delta = 2$), a compatible graph, a realignment moment with $\abs{M_1} > \abs{M_2}, \abs{M_3} $, a realignment moment with $\abs{M_3} \geq \abs{M_1}, \abs{M_2} $, a joint realignment moment, and a joint realignment moment with subsets of size one (\ie colors) that alternate along the cycle.}
    \label{fig:Graphs-ex}
\end{figure}

The degree of compatibility $\Delta$ --  or equivalently $\mF$ through Eq.~\eqref{eq:Delta_0} -- can be computed for the colored graphs introduced above, as summarized below, see Ref.~\cite{Carrozza:2026qcf} for more details. 

\begin{prop} \label{prop:families-low-delta}
The degree of compatibility takes the following values for the colored graphs above:
    \begin{enumerate}
        \item $\Delta = 0$: the cyclic graphs associated to a subset $M\subset \paa{1,\dots,D}$ with $\abs{M} = 1$ (see \eg Ref.~\cite{Collins:2024pip} and references therein), a subset of the maximally single-trace graphs (see Ref.~\cite{Carrozza:2026qcf}) and the joint realignment moments with subsets of colors alternating along the cycle (see the bottom-rightmost graph of Fig.~\ref{fig:Graphs-ex} and Ref.~\cite{Lionni2018}) are compatible for any $D\ge 3$.
        \item $\Delta = 1$: the realignment moments with $\abs{M_1} = \abs{M_2} = 1$ have a degree of compatibility equal to one for any $D \geq 3$ (see Ref.~\cite{Lionni2018}).
        \item Arbitrary $\Delta$: the degree of compatibility of a cyclic graph $G$ with $2k$ vertices and associated to a subset $M \subset \paa{1,\dots,D}$ with $\abs{M} \leq \lfloor D/2 \rfloor$, can be computed explicitly (see Refs.~\cite{Bonzom:2015axa, Bonzom_Lionni_Rivasseau_2017,AIHPD_2015__2_1_1_0,bonzom2014tensormodelsviewpointmatrix}): 
        \begin{equation}
            \Delta(G) = \frac{\abs{M}(\abs{M}-1)}{2}(k-1) \,.
        \end{equation}
        The degree of compatibility of the realignment moments for arbitrary value of $\abs{M_1}$,  $\abs{M_2}$ and $\abs{M_2}$ can also be explicitly computed (see \eg App.~C of Ref.~\cite{Carrozza:2026qcf} and Ref.~\cite{Lionni2018}).
    \end{enumerate}
\end{prop}

For some of the colored graphs above, one can determine the graphs in $\MD{}^{\conn}(\gG)$, that is, the connected  graphs that dominate the expansion in Eq.~\eqref{eq:Exp_gG}. All the examples studied in the literature so far have tree-like dominant pairings.  Below is a non-exhaustive summary regarding this particular aspect (among other things, the references cited also characterize the graphs belonging to $\MD{}^{\conn}(\gG)$ and provide large $p$ results concerning $\mu_{\gG}^{\conn}$). 

\begin{theo}
\label{thm:families-tree-like-LO}
The tree-like dominance is known in the following cases:
    \begin{enumerate}
        \item If $\gG$ consists in copies of connected melonic graphs, then $\gG$ has only tree-like dominant pairings (see Ref.~\cite{PhysRevD.85.084037}).
        
        \item If $D$ is even and $\gG$ consists in copies of cyclic graphs, then $\gG$ has tree-like dominant pairings; furthermore, if $\gG$ consists in copies of cyclic graphs associated to subsets $M \subset \paa{1,\dots,D}$ with $\abs{M} < D/2$, then $\gG$ has only tree-like dominant pairings (see Refs.~\cite{Bonzom:2015axa, Bonzom_Lionni_Rivasseau_2017}). 
        
        \item If $\gG$ consists in copies of a single kind of realignment or joint realignment moments with $\abs{M_3} \geq \abs{M_1},\abs{M_2}$ (see Ref.~\cite{Bonzom_Lionni_Rivasseau_2017, Bonzom_Lionni_2017, LIONNI2019600, Lionni2018}), then it has only tree-like dominant pairings. 
        
        \item Any $\gG$ consisting of connected planar $3$-colored graphs has only tree-like dominant pairings (see Ref.~\cite{Bonzom:2018btd}). 
    \end{enumerate}
\end{theo}

Note that, while with the exception of Ref.~\cite{Lionni2018}, the references listed in the theorem concern $\gG=\{G_1 , \ldots , G_p\}$ where the $G_i$ are \emph{connected} graphs, by Lem.~\ref{lem:tree_like_B_implies_G}, the same results -- having tree-like or only tree-like dominant pairings --  hold even when the graphs $G_i$ are not necessarily connected. Lem.~\ref{lem:tree_like_B_implies_G} also states that if a graph $G_i$ is non-connected, then in any graph $\widehat G$ of $\MD{}^{\conn}(\gG)$, the connected components of $G_i$ belong to different connected components of $\widehat G$, and the graph $G_i$ is the only connection point between these connected components. This generalizes the following fact known in the literature on matrix models with multi-trace interactions: for $D=2$, if the graphs $G_i$ are connected, then $\MD{}^{\conn}(\gG)$ consists of planar ribbon graphs, but if the graphs $G_i$ are not connected, then  $\MD{}^{\conn}(\gG)$ consists of ``trees of planar ribbon graphs'', or ``trees of spheres'', see \eg page 14 of Ref.~\cite{Alvarez-Gaume:1992idg}.

The following statement is proved by the use of Thm~6.4 in Ref.~\cite{Carrozza:2026qcf} and relies on flips of colors different from $0$ (see above Lem.~\ref{lem:faces-when-flip}). 

\begin{theo} \label{thm:families-any-delta}
    For any $\delta\in\bb{N}^\star$, there are infinitely many graphs $G\in\cG_D^{\conn}$ with $\Delta(G)=\delta$. Equivalently, there are infinitely many $G\in\cG_{D}^{\conn}$ with $\mF(G)=\frac D 2 k(G) + \frac{F(G)}{D-1} - \frac{2\delta}{D-1}$. 
\end{theo}

\begin{proof}
    Consider two graphs $G_1, G_2\in\cG_{D}^{\conn}$. Let $c\in\{1, \ldots, D\}$ and select a pair of edges of color $c$ in $G_1$ and $G_2$, respectively. Flipping these two edges produces a new connected graph $G$. As shown in Ref.~\cite{Carrozza:2026qcf}, assuming that $\gG = \paa{G_1,G_2}$ has tree-like dominant pairings or satisfies the condition $\Delta(G_1) + \Delta(G_2) < \frac{(D-1)(D-2)}{2}$, the following relation necessarily holds:
    \begin{equation}\label{eq:delta_flips}
        \Delta(G)=\Delta(G_1)+\Delta(G_2)\,, \quad \rm{or equivalently,} \quad  \mF(G)=\mF(G_1)+\mF(G_2) \,. 
    \end{equation}
    
    Now, let $D \geq 3$ and $B$ be a realignment moment with $\abs{M_1} = \abs{M_2} = 1$. By Prop.~\ref{prop:families-low-delta} and Thm.~\ref{thm:families-tree-like-LO}, we know that $\Delta(B) = 1$ and $B$ admits only tree-like dominant pairings. Applying the result from Ref.~\cite{Carrozza:2026qcf} we just recalled, one can construct a connected $D$-colored graph $H$ by performing $\delta-1$ flips on $\delta$ copies of $B$, yielding: $\Delta(H)=\delta$. 
    
    Furthermore, this construction can be extended to infinite families of graphs while maintaining a compatibility degree of $\delta$. If $\delta < \frac{(D-1)(D-2)}{2}$, flips can be used to incorporate any number of compatible trace-invariants, thanks to Eq.~\eqref{eq:delta_flips}. If $\delta > \frac{(D-1)(D-2)}{2}$, one may proceed similarly with the additional requirement that any compatible trace-invariant incorporated in the graph must also have tree-like dominant pairings. Since \eg $2$-vertex graphs are compatible and have tree-like dominant pairings, this is sufficient to generate an infinite family of graphs with degree of compatibility equal to $\delta$. 
\end{proof}

As a concrete example, we present in Fig.~\ref{fig:graph_Delta_4} a $5$-colored graph with a degree of compatibility equal to four. The figure highlights four realignment moments, each satisfying  $\abs{M_1} = \abs{M_2} = 1$, represented by the four ``circles'' or ``wheels''. Following the proof of Thm.~\ref{thm:families-any-delta}, the degree of compatibility of a graph constructed via four flips of realignment moments and further decorated with flips of $2$-vertex graphs -- as illustrated in Fig.~\ref{fig:graph_Delta_4} -- is $\Delta = 4$.

\begin{figure}[ht]
    \centering
    \includegraphics[width = .85\textwidth]{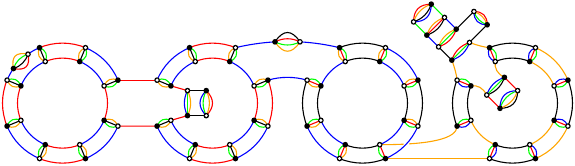}
    \caption{Example of a connected $5$-colored graph with a degree of compatibility: $\Delta = 4$.}
    \label{fig:graph_Delta_4}
\end{figure}

Actually, graphs with $k(G)=k$ built in this way from a fixed collection of $r>1$ building blocks can be put in one-to-one correspondence with decorated unrooted plane trees, so that their number grows  exponentially with $k$.\footnote{More precisely, we expect this number to be equivalent to $c k^{-5/2}z^k$ in the large $k$ limit, with $c >0$ and $z>1$ (see \eg Ref.~\cite{stufler2026probabilistic}).}

\section{First counterexample to large \texorpdfstring{$N$}{N} factorization in complex random tensors} \label{sec:counterex}

A first counterexample to the large-$N$ factorization of trace-invariants  in the context of real random tensors was proposed in Ref.~\cite{Berthold:2026zxk}, thus illustrating the result established probabilistically in Ref.~\cite{Gurau:2025evo}. While the proof given in Ref.~\cite{Gurau:2025evo} naturally extends to the complex case, no counterexample is known in the literature to date. Here, we present an example of a trace-invariant $H$ for which the large-$N$ factorization fails. More concretely, for $H$ the graph with $D=6$ represented on the left-hand side of Fig.~\ref{fig:MST_incomp}, we have:
\begin{equation}
    \mean{\tr_H(T,\bar{T}) \cdot \tr_{\bar{H}}(T,\bar{T})} \underset{N \to \infty}{\sim} \mean{\tr_H(T,\bar{T}) \cdot \tr_{\bar{H}}(T,\bar{T})}_\rm{conn} \,.
\end{equation}
Indeed, a numerical analysis shows that $\mF(H) = 26 < 27 = D k(H)/2$. Hence, the first statement of Lem.~\ref{lem:conditions_Gurau} implies the non-factorization.

\begin{figure}[ht]
	\centering
	\includegraphics[height = 6cm]{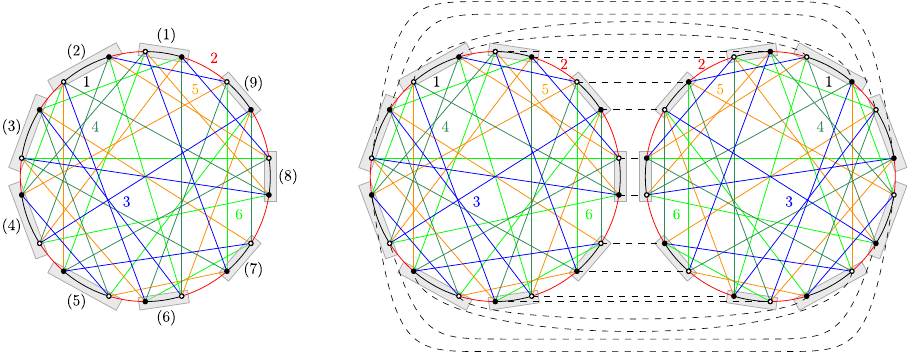}
	\caption{Left: an example of incompatible, maximally single-trace $6$-colored graph. Using the vertex labeling in the figure, we have: $\sigma_1 = \id$, $\sigma_2 = (1\, 2\, 3\, 4\, 5\, 6\, 7\, 8\, 9)$, $\sigma_3 = (1\, 5\, 4\, 9\, 2 \,7\, 6\, 8\, 3)$, $\sigma_4 = (1\, 7\, 4\, 3 \,9 \,5\, 2 \,8\, 6)$, $\sigma_5 = (1 \,8 \,2 \,9\, 7\, 5\, 3\, 6\, 4)$ and $\sigma_6 = (1\, 5\, 9\, 7 \,6\, 2\, 4\, 8\, 3)$. Numerically, we found $\mF(H) = 26$, \ie $\Delta(H) = 10$. Right: one of the dominant pairing $\widehat{H \sqcup \bar H} \in \MD{}(H \sqcup \bar{H})$, \ie a $7$-colored graph that realizes the bound $F_0(\widehat{H \sqcup \bar H}) = D k(H) = 54 > 52 = \mF(H) + \mF(\bar H)$.}
	\label{fig:MST_incomp}
\end{figure}

A clarifying comment is in order, in relation to the results of Ref.~\cite{Ferrari:2017jgw} and the definition of tree-like pairing we are relying on in the present work. Let $H$ denote an arbitrary maximally single-trace graph, and let $\gH = \{ \underbrace{H, \ldots , H}_{p \; \mathrm{times}} , \underbrace{\bar H , \ldots , \bar H}_{p \; \mathrm{times}}\}$. In Ref.~\cite{Ferrari:2017jgw}, it was proven that a certain infinite subfamily of connected graphs $\mathcal{M}_{\mathrm{melo}}$ in $\MD{}^{\conn}(\gH)$ (hence, that contributes to $\mean{ \left( \tr_H (T, \bar T ) \tr_{\bar{H}} (T , \bar T)\right)^p }_{\mathrm{conn}}$ at leading order) can be put in one-to-one correspondence with trees. One could therefore expect this result to imply that $\gH$ has tree-like dominant pairings, in the sense of Def.~\ref{def:DEF}. However, this would lead to an inconsistency between the counterexample illustrated in Fig.~\ref{fig:MST_incomp} and our third main theorem, Thm.~\ref{main-theo3}. The resolution of this apparent inconsistency is that $\gH$ does not in fact admit tree-like dominant pairings in the sense of Def.~\ref{def:DEF}. Indeed, in the tree-like combinatorial structure of a graph $\widehat{H}\in \mathcal{M}_{\mathrm{melo}}$ -- which is sometimes known as a generalized melonic graph -- a ``vertex'' of the tree is not an individual building block appearing in $\gH$, but rather a \emph{bilocal} object of the form $H \sqcup \bar H$. By contrast, Def.~\ref{def:DEF} requires any ``vertex'' of the combinatorial tree underlying a tree-like graph $\widehat H$ to be \emph{local} to some $H$ or $\bar H$. This observation provides yet  another crucial difference between two interesting families of tensor models, that are both referred to as \emph{melonic} in the literature: the melonic theories dominated by \emph{local} melonic corrections on the one hand \cite{PhysRevD.85.084037}, which in many ways behave like vector models; and, on the other hand, the melonic theories dominated by \emph{bilocal} melonic corrections \cite{Gurau:2010ba, Bonzom:2011zz, Carrozza:2015adg, Ferrari:2017jgw}, which generate more interesting physics, and can in particular reproduce the main features of the strongly-coupled Sachdev-Ye-Kitaev model \cite{Witten:2016iux, Klebanov:2016xxf}.

\

For maximally single-trace graphs, the condition of Lem.~\ref{lem:conditions_Gurau} is, in fact, also necessary for the non-factorization to hold:
\begin{lem}
\label{lem:non-facto-MST}
 Let $D \geq 3$ and $H \in \cG_D^{\conn}$ be a maximally single-trace graph. Then, the following conditions are equivalent: 
 \begin{enumerate}
     \item $\mF(H)\le \frac D 2 k(H)$, or equivalently $\Delta(H) \geq \frac{D(D-1)}{4}$.
     \item $\{H, \bar H\}$ does not factorize at large $N$.
     \item $\mF(H \sqcup \bar H) = Dk(H)$.
 \end{enumerate}
Furthermore, the inequalities of condition {\it 1} are strict if and only if $\mean{\tr_{H \sqcup \bar H}(T,\bar T)} \underset{N \to \infty}{\sim} \mean{\tr_{H \sqcup \bar H}(T,\bar T)}_{\mathrm{conn}}$.
\end{lem}

\begin{proof}
    \textit{1.} $\Rightarrow$ \textit{2.}: this implication is given by point \textit{1.} of Lem.~\ref{lem:conditions_Gurau}.
        
    \textit{2.} $\Rightarrow$ \textit{3.}: we let $G=H \sqcup \bar H$ and assume that $\{H, \bar H\}$ does not factorize at large $N$, so that $\mF(G) = \mFc(G)$. Applying the inequality of Lem.~\ref{lem:Razvans-bound-on-F}, one therefore has: 
    \begin{equation}
    \label{eq:upp-bound-mFH-MST}
        \mF(G)  \leq \frac D 2 \cdot 2 k(H) + \frac{1}{D-1} \cdot 2 \cdot \frac{D(D-1)}{2} - D(2 - 1) = Dk(H)\,.
    \end{equation}
    Now, since adding an edge of color 0 between every vertex of $H$ and its mirror in $\bar H$ produces a graph $\widehat{G}$  with $F_0(\widehat{G}) = Dk(H)$, one also  has that $\mF(G) \geq Dk(H)$. Combining both inequalities leads to $\mF(G) = D k(H)$.
    
    \textit{3.} $\Rightarrow$ \textit{1.}: we suppose that $\mF(G) = Dk(H)$. The  $(D+1)$-colored graph $\widehat{G }$ introduced above is connected and satisfies $F_0(\widehat{G }) = Dk(H)$. Hence, $\paa{H, \bar H}$ does not factorize at large $N$, thus implying that $2 \mF(H) \leq \mF( G) = Dk(H)$.
\end{proof}

\

The following result shows that any maximally single-trace graph $H$ such that $\{H, \bar H\}$ does not factorize at large $N$ provides an example such that the leading terms of all cumulants in the expansion, Eq.~\eqref{eq:mom-cum-formula}, contribute to the leading term of $\mean{\bigl(\tr_G(T,\bar T)\bigr)^p}$, \ie if $\mathbf{G} = \paa{G_1,\dots,G_p}$ with $G_i= G = H \sqcup \bar H$ for any $i \in \{1, \ldots, p\}$, we get:
\begin{equation}
\label{eq:result-mst-nonfacto}
    \mean{\bigl(\tr_G(T,\bar T)\bigr)^p} \underset{N \to \infty}{\sim}
   \sum_{\pi \in \cal{P}(p)} \prod_{B \in \pi} \mean{\paa{\tr_{G_i}(T,\bar T)}_{i \in B}}_\mathrm{conn}\,,
\end{equation}
 where the label ``$\mathrm{conn}$'' refers to connectedness relative to $\mathbf{G}$.\footnote{In particular, we recall that, while any graph $\widehat G_B \in \cG_{D+1}^{\conn}(\{ G_i \}_{i \in B})$ contributing to $\mean{\paa{\tr_{G_i}(T,\bar T)}_{i \in B}}_\mathrm{conn}$ is not necessarily connected, it must be such that $K(\widehat G_B ; \{ G_i \}_{i \in B})$ is connected.}
More precisely, we get the following. 
    
\begin{prop} 
\label{prop:all-cum-LO}
    Let $D \geq 3$, $p \in \bb{N}^*$, $H \in \cG_D^{\conn}$ a maximally single-trace graph and define $G = H \sqcup \bar H$. If $\{H, \bar H\}$ does not factorize at large $N$, then one has the following large $N$ limits for the moments and cumulants of $\tr_G(T,\bar T)$:
    \begin{align}
    \label{eq:asymptotics-cumulants-HH}
      &\lim_{N\rightarrow \infty}  \mean{\paa{N^{\frac D 2 k(G)}\tr_{G_i}(T,\bar T)}_{1\le i \le p}}_\mathrm{conn} = \mu_{\paa{G_i}_{1\le i \le p}}^{\conn}\;,\\
      & \lim_{N\rightarrow \infty}  \mean{\Bigl(N^{\frac D 2 k(G)}\tr_G(T,\bar T)\Bigr)^p} =  \sum_{\pi \in \cal{P}(p)} \prod_{B \in \pi} \mu_{\paa{G_i}_{i \in B}}^{\conn}\;.
        \label{eq:asymptotics-moments-HH}
    \end{align}
   If the sequence $\{\mu_{\paa{G_i}_{1\le i \le p}}^{\conn}\}_{p\in \mathbb{N}^\star}$ is the cumulant sequence of a unique  real random variable $x$, then the random variable $N^{\frac D 2 k(G)}\tr_G(T,\bar T)$ converges in distribution to  $x$.  
\end{prop}

\proof With the same notations, we let $\tilde G=\bigsqcup_i G_i$. The contributions in the sum of Eq.~\eqref{eq:mom-cum-formula} that correspond to a given partition $\pi$ can equivalently be seen as the contributions from the graphs in $\cG_{D+1}(\tilde G)$ for which the connected components relative to $\mathbf{G}$  induce the partition $\pi\in \mathcal{P}(p)$. The proof of the proposition relies on two steps: we first show that any graph $\widehat I\in \cG_{D+1}(\tilde G)$ satisfies $F_0(\widehat I)\le pDk(H)$, and we then  exhibit, for every $\pi$, a contribution $\widehat I\in \cG_{D+1}(\tilde G)$  satisfying $F_0(\widehat I)=pDk(H)$, thereby proving the statement. 

Considering  $\widehat I\in \cG_{D+1}(\tilde G)$, we label the copies $H$'s and $\bar H$'s from $1$ to $2p$  and let $\eta \in \cal{P}(2p)$ be the partition induced by the connected components of $\widehat I$ (in the usual sense): one has $ \widehat I = \bigsqcup_{B \in \eta} \widehat I_B $, where the graphs $\widehat I_B$ are connected. We then rewrite the number of color-$0$ faces of $\widehat I$ as: 
   \begin{equation}
       F_0(\widehat I) = \sum_{B \in \eta} F_0(\widehat I_B) = \sum_{\substack{B \in \eta \\ \abs{B} = 1}} F_0(\widehat I_B) + \sum_{\substack{B \in \eta \\ \abs{B} \geq 2}} F_0(\widehat I_B)  \leq \sum_{\substack{B \in \eta \\ \abs{B} = 1}} \mF(H)  + \sum_{\substack{B \in \eta \\ \abs{B} \geq 2}} \mFc(\mathbf{I_B}) \,,
   \end{equation}
   where $\mathbf{I_B}$ is the collection of copies of $H$ and $\bar H$ whose labels are in $B$ (the connected components of the graph obtained from $\widehat I_B$ by removing  the color-$0$ edges).
   Since by assumption $\{H, \bar H\}$ does not factorize at large $N$, we know from Lem.~\ref{lem:non-facto-MST} that $\mF(H)\le Dk(H)/2$. We use Lem.~\ref{lem:Razvans-bound-on-F} to bound the terms 
   for the blocks $B$ of $\eta$ with $\abs{B} \geq 2$ by $\frac D 2 k(G)\lvert B\rvert + \frac D 2\lvert B \rvert - D(\lvert B\rvert -1)$, thus getting the bound: 
   \begin{equation}
   \label{eq:bound-mst-faces}
       F_0(\widehat I) \leq \sum_{\substack{B \in \eta \\ \abs{B} = 1}} \frac{D}{2} k(H)  + \sum_{\substack{B \in \eta \\ \abs{B} \geq 2}} \pac{\frac D 2 (k(H) - 1) \abs{B}  + D} = p D k(H) -  \frac D 2 \sum_{\substack{B \in \eta \\ \abs{B} \geq 2}} (\abs{B} - 2) \leq  p D k(H) \,,
   \end{equation}
   where the last equality comes from the relation $2p = \displaystyle\sum_{B \in \eta} \abs{B}$.
   This shows that any $\widehat I\in \cG_{D+1}(\tilde G)$ satisfies  $F_0(\widehat I)\le pDk(H)$, where the inequality is strict whenever $\eta$ contains a block with three or more elements.
   
   For any $\pi\in \mathcal{P}(p)$, we can easily build a graph in $\cG_{D+1}(\tilde G)$ which saturates this bound. 
   Considering the $p$ labeled copies $G_i=H\sqcup \widehat H$ and a permutation $\sigma\in S_p$ whose cycles have supports corresponding to the blocks of $\pi$, we add  for every $1\le i\le p$ some edges of color 0 between the  vertices of the $H$ belonging to $G_i$ with their mirror vertices in the copy of $\bar H$ belonging to $G_{\sigma(i)}$, and denote by $\widehat J_\sigma$ the resulting graph. Two examples are illustrated in  Fig.~\ref{fig:all_cum}, corresponding to the partitions $0_p$ and $1_p$. The  connected components relative to $\mathbf{G}$ of $\widehat J_\sigma$ induce the partition $\pi\in \mathcal{P}(p)$. Moreover, one has  $F_0(\widehat J_\sigma) = pDk(H)$. 
   
   Together with Eq.~\eqref{eq:partial_cum}, this shows that the cumulants $ \mean{\paa{\tr_{G_i}(T,\bar T)}_{1\le i \le p}}_\mathrm{conn}$ scale as $N^{ - \frac{pD}{2} k(G) }$, which proves  Eq.~\eqref{eq:asymptotics-cumulants-HH}. All terms in the expansion Eq.~\eqref{eq:mom-cum-formula} of $\mean{\Bigl(N^{\frac D 2 k(G)}\tr_G(T,\bar T)\Bigr)^p}$ therefore also scale as $N^{\frac{pD}2k(G)}$. This proves Eq.~\eqref{eq:asymptotics-moments-HH}. 
   \qed 
 
    \begin{figure}[ht]
    	\centering
    	\includegraphics[height = 4cm]{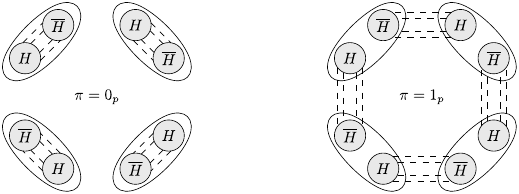}
    	\caption{Examples of $\widehat J_\sigma$ with $\pi = 0_p$ (left) and $\pi = 1_p$ (right).}
    	\label{fig:all_cum}
    \end{figure}

\ 

\begin{ex}
    The graph $H$, shown on the left in Fig.~\ref{fig:MST_incomp}, satisfies the conditions of Prop.~\ref{prop:all-cum-LO}, thereby ensuring that all cumulants contribute to the dominant term of $\mean{\bigl(\tr_G(T,\bar T)\bigr)^p}$ for any $p$, and that:
    \begin{equation}
       - \ln \mean{\bigl(\tr_G(T,\bar T)\bigr)^p} \underset{N \to \infty}{\sim} 54 p \ln N \,.
    \end{equation}
\end{ex}

We provide below two situations for which the asymptotic distribution is determined explicitly. 

\begin{prop}\label{prop:all-cum-LO-distr} 
    Let $D \geq 3$, $p \in \bb{N}^*$, $H \in \cG_D^{\conn}$ a maximally single-trace graph and define $G = H \sqcup \bar H$. If  $\mF(H)< \frac D 2 k(H)$, or equivalently $\mean{\tr_G(T,\bar T)} \underset{N \to \infty}{\sim} \mean{\tr_G(T,\bar T)}_{\mathrm{conn}}$, then: 
    \begin{itemize}
   \item If $\mF^{\conn}(H\sqcup H)<\mF^{\conn}(H\sqcup \bar H)$, then the random variable  $N^{\frac{D}{2}k(G)}\tr_G(T,\bar T)$ converges in distribution to an exponential random variable of parameter $1/\mu^{\conn}_G$ ({\it i.e.}~a real random variable supported on $[0, \+\infty)$ and of probability density function $x \mapsto \frac 1 {\mu^{\conn}_G}\mathrm{e}^{- x / {\mu^{\conn}_G}}$):
   \begin{equation}
          \label{eq:exp-rr-mst}
      \lim_{N\rightarrow \infty} \mean{\Bigl\{N^{\frac D 2 k(G)}\tr_{G_i}(T,\bar T)\Bigr\}_{1\le i \le p}}_{\mathrm{conn}} =  (p-1)! (\mu^{\conn}_G)^p, \qquad \lim_{N\rightarrow \infty} \mean{\Bigl(N^{\frac D 2 k(G)}\tr_G(T,\bar T)\Bigr)^p} = p! (\mu^{\conn}_G)^p\;.
    \end{equation}
    \item If $\mF^{\conn}(H\sqcup H)=\mF^{\conn}(H\sqcup \bar H)$ and $\mu^{\conn}_G = \mu^{\conn}_{H\sqcup H}$ (which hold in particular if $H=\bar H$), then the random variable $N^{\frac{D}{2}k(G)}\tr_G(T,\bar T)$ converges in distribution to a Gamma random variable of parameters $(\frac 1 2, 2\mu^{\conn}_G)$ ({\it i.e.}~the real random variable supported on $]0,\infty)$  and of probability density function $x \mapsto \frac 1 {\sqrt{x \mu^{\conn}_G \pi}}\mathrm{e}^{- x / {\mu^{\conn}_G}}$):
       \begin{equation}
          \label{eq:exp-rr-mst2}
          \lim_{N\rightarrow \infty} \mean{\Bigl\{N^{\frac D 2 k(G)}\tr_{G_i}(T,\bar T)\Bigr\}_{1\le i \le p}}_{\mathrm{conn}}\hspace{-0.2ex} =  \frac 1 2 (p-1)! (2\mu^{\conn}_G)^p, \quad
      \lim_{N\rightarrow \infty} \mean{\Bigl(N^{\frac D 2 k(G)}\tr_G(T,\bar T)\Bigr)^p} = (2p-1)!! (\mu^{\conn}_G )^p,
    \end{equation}
    where $(2p-1)!!=\frac{(2p)!}{2^p p!}$.
    \end{itemize}
\end{prop}

\proof    We assume that $\mF(H)< \frac D 2 k(H)$ and $\mF^{\conn}(H\sqcup H)<\mF^{\conn}(H\sqcup \bar H)$. With the same notations as the proof of the last proposition, in  Eq.~\eqref{eq:bound-mst-faces}, for the inequality to be saturated, $\eta$ must only have blocks of size two, containing a copy of $H$ and a copy of $\bar H$. This shows that the elements of $\MD{}^{\conn}(\gG)$ are of the form $\widehat J_\sigma$ for $\sigma$ a cycle of length $p$, or a graph of similar form where the restriction to the copy of $H$ belonging to $G_i$ and the copy of $\bar H$ belonging to $G_{\sigma(i)}$ can be connected in any way which attains the maximal value $\mF(G)$. There are $(p-1)!(\mu^{\conn}_G)^p$ such graphs, which shows that the cumulants and moments of $N^{\frac{D}{2}k(G)}\tr_G(T,\bar T)$ converge to those of an exponential random variable of parameter $1/\mu^{\conn}_G$, given in Eq.~\eqref{eq:exp-rr-mst}.  
   
Assuming now $\mF(H)< \frac D 2 k(H)$ and $\mF^{\conn}(H\sqcup H)=\mF^{\conn}(H\sqcup \bar H)$ as well as $\mu^{\conn}_G = \mu^{\conn}_{H\sqcup H}$, $\eta$ must only have blocks of size two, but it may contain two copies of $H$ or two copies of $\bar H$. For every $\sigma$ consisting of a cycle of length $p$, one has to choose for each block of $\eta$ whether   it consists of two copies of $H$ or  $\bar H$ (one belonging to $G_i$ and the other one to $G_{\sigma(i)}$), or a copy of $H$ and a  copy of $\bar H$. Due to the cyclic structure, we only have this freedom for $p-1$ blocks. There are in total $(p-1)!2^{p-1}(\mu^{\conn}_G)^p$ such graphs, which shows that the cumulants and moments of $N^{\frac{D}{2}k(G)}\tr_G(T,\bar T)$ converge to those of a Gamma distribution of parameters $(\frac 1 2, 2\mu^{\conn}_G)$.\qed 
   
   \

One can in principle characterize the asymptotic distribution in the other cases not considered in the previous proposition. For instance, if $\mF(H)= \frac D 2 k(H)$, then, in addition to the graphs with a cyclic structure corresponding to blocks of $\eta$ of size two, one must consider the possibility for $\eta$ to be of size one. In this case, there are also contributions to the cumulants that have a chain structure that starts and ends on a block of $\eta$ of size one. If for instance $\mF^{\conn}(H\sqcup H)<\mF^{\conn}(H\sqcup \bar H)$, the cumulant therefore appears to be $(p-1)!(\mu^{\conn}_G)^p +  p! (\mu^{\conn}_G)^{p-1}(\mu^{\conn}_H)^2$. If $\mF(H)= \frac D 2 k(H)$, $\mF^{\conn}(H\sqcup H)=\mF^{\conn}(H\sqcup \bar H)$ and $\mu^{\conn}_G \neq \mu^{\conn}_{H\sqcup H}$, we find  the first four cumulants to be $\mu^{\conn}_G$ for $p=1$, $(\mu^{\conn}_G)^2 + (\mu^{\conn}_{H^2})^2$ for $p=2$, $2[(\mu^{\conn}_G)^3 + 3 \mu^{\conn}_G (\mu^{\conn}_{H^2})^2]$ for $p=3$, and $3![(\mu^{\conn}_G)^4 + 6 (\mu^{\conn}_G)^2 (\mu^{\conn}_{H^2})^2 + (\mu^{\conn}_{H^2})^4]$ for $p=4$, and so on. We leave these other cases for future work.

\section{Large \texorpdfstring{$N$}{N} factorization of families of trace-invariants} \label{sec:largeN}

\paragraph{First main theorem.}The definition of large $N$ factorization of a graph over its connected components was given in Def.~\ref{def:facto}. Our first main theorem provides a sufficient combinatorial bound for large $N$ factorization to hold, in a similar spirit as Lem.~\ref{lem:conditions_Gurau} but with notable differences.
\begin{theo}
\label{main-theo1}
Let $D \geq 2$ and $G \in\cG_{D}$. We denote the connected components of $G$ by $G_1,\ldots,G_{\kappa}$. If the following bound holds:
\begin{equation}
\label{eq:assumpt-2}
    \sum_{i=1}^{\kappa}\mF(G_i) > \frac {D} 2 k(G)  + \frac{F(G)}{D-1}  - D 
    \,,
\end{equation}
then $G$ factorizes over its connected components at large $N$.
\end{theo}

\proof 
    We first show that the \apriori weaker assumption
    \begin{equation}
    \label{eq:assumpt-1}
        \mF(G) > \frac {D} 2 k(G) + \frac{F(G)}{D-1}  - D \;,
    \end{equation}
    implies that $G$ factorizes over its connected components. 
    Assuming  Eq.~\eqref{eq:assumpt-1}, we consider a $\widehat G\in \cG_{D+1}(G)$ such that $F_0(\widehat G) > \frac {D} 2 k(G) + \frac{F(G)}{D-1}  - D $. Lem.~\ref{lem:Razvans-bound-on-F} gives an upper bound for  $F_0(\widehat G) $, so that:
    \begin{equation}
        \frac {D} 2 k(G)  + \frac{F(G)}{D-1}  - D  <  F_0(\widehat G)   \le \frac D 2 k(G) + \frac{F(G)} {D-1}  -D\bigl(\kappa(G) - \kappa(\widehat G)\bigr)\;, 
    \end{equation}
    and in particular, since  $\kappa(G) \ge \kappa(\widehat G)$ (adding the edges of color 0 can only decrease the number of connected components): 
    \begin{equation}
        0\le \kappa(G) - \kappa(\widehat G) <  1\;.  
    \end{equation}
    Since $\kappa(G) - \kappa(\widehat G)$ is an integer, it must vanish, so that $\widehat G$ is a disjoint union $\widehat{G_1}\sqcup \ldots \sqcup \widehat{G_\kappa}$, where $\widehat{G_{i}}\in \cG_{D+1}(G_i)$ for any $i$. Any graph $\widehat G$ such that   $F_0(\widehat G) = \mF(G)= \max_{\widehat G \in \cG_{D+1}(G)}F_0(\widehat G)$ is therefore necessarily of this form:  $G$ factorizes over its connected components. 

    From Eq.~\eqref{eq:exp-and-asympt-moments} and Eq.~\eqref{eq:mom-cum-formula}, it is always true that $\mF(G) \ge  \sum_{i=1}^{\kappa}\mF(G_i)$. As a consequence, the condition Eq.~\eqref{eq:assumpt-1} is necessarily satisfied if Eq.~\eqref{eq:assumpt-2} holds. \qed

\ 

It may \apriori be possible for Eq.~\eqref{eq:assumpt-1} to be true while Eq.~\eqref{eq:assumpt-2} does not hold, that is, it may \apriori  be possible to have  $\mF(G) > \frac {D} 2 k(G)  + \frac{F(G)}{D-1}  - D \ge  \sum_{i=1}^{\kappa}\mF(G_i)$. However, we have just shown that the leftmost inequality implies that $G$ factorizes over its connected components at large $N$, which from Eq.~\eqref{eq:mF-equal-parts} implies that  $\mF(G) = \sum_i \mF(G_i)$, which implies Eq.~\eqref{eq:assumpt-2}. \emph{The two conditions Eq.~\eqref{eq:assumpt-2} and Eq.~\eqref{eq:assumpt-1} are therefore equivalent.} In practice, however, given a collection of connected colored graphs $G_i$, it is a lot less costly to verify Eq.~\eqref{eq:assumpt-2} rather than Eq.~\eqref{eq:assumpt-1}.

\begin{rem}
    The two equivalent conditions Eq.~\eqref{eq:assumpt-2} and Eq.~\eqref{eq:assumpt-1} can be written in terms of the degree of compatibility. Indeed, they are respectively equivalent to:
    \begin{equation}\label{eq:assumpt_deg_comp}
        0 \leq \sum_{i = 1}^\kappa \Delta(G_i) < \frac{D(D-1)}{2} \,, \qquad \mathrm{and}
        \qquad 0 \leq \Delta(G) < \frac{D(D-1)}{2} \,.
    \end{equation}
    Let us fix $\vec \sigma \in S_{k(G)}^D(G)$ and let $\nu \in S_{k(G)}$ be an dominant pairing of $G$. Denoting by $d(\cdot,\cdot)$ the Cayley distance on the symmetric group, the correspondence between non-labeled graphs and permutations yields the following formula for the degree of compatibility:
    \begin{equation}
        \Delta(G) = \sum_{1 \leq i<j \leq D} \frac 1 2 \pac{ d(\sigma_i,\nu) + d(\nu, \sigma_j) - d(\sigma_i,\sigma_j)} \,.
    \end{equation}
    Thus, the conditions of Eq.~\eqref{eq:assumpt_deg_comp} have a geometrical interpretation. Specifically, Eq.~\eqref{eq:assumpt_deg_comp} imposes that: if $\nu$ is an dominant pairing of $G$, then there exist at least two distinct integers $i,j \in \paa{1,\dots,D}$ such that 
    \begin{equation}
        d(\sigma_i,\nu) + d(\nu, \sigma_j) = d(\sigma_i,\sigma_j) \,,
    \end{equation}
    \ie at least one of the triangle inequalities is saturated.
\end{rem}

In the second statement of Lem.~\ref{lem:conditions_Gurau}, the bound that each connected component $G_i$ must satisfy is 
\begin{equation}
\label{eq:bound-razvan-cc}
        \mF(G_i) > \frac D 2 k(G_i), 
\end{equation}
but this bound must also  be satisfied by the boundary graphs $\partial \mathring G_i $ of all the $\mathring G_i\in\cG_{D+1}^\circ(G_i)$. 
This condition becomes unpractical whenever at least one $G_i$ has more than a handful of vertices. Indeed, a single copy of $G_i$  can in principle generate boundary graphs with very different properties, and in particular, arbitrarily small $\mF$; for instance: a planar $G_i$ can \apriori generate boundary graphs of arbitrary genus (see the discussion on page 9 of Ref.~\cite{Gurau:2025evo}); a $G_i$ consisting of a collection of melonic graphs each with $k=2$ can generate colored graphs of any kind (see Thm.~2 of Ref.~\cite{Bonzom_Lionni_Rivasseau_2017} and Thm.~3.4.2 of Ref.~\cite{Lionni2018}). 

Focusing now on the condition of Eq.~\eqref{eq:assumpt-2}, the lower bound that one must verify for $\displaystyle \sum_{i=1}^p\mF(G_i) $ (where $G_i$ are the connected components of $G$) is larger or equal to the bound given by Eq.~\eqref{eq:bound-razvan-cc}. Indeed, since $F(G)\ge \frac{D(D-1)}{2}\kappa$ with equality if and only if all the connected components are maximally single-trace graphs, one has 
\begin{equation}
    \frac{F(G)}{D-1}  - D  \ge D\Bigl(\frac{\kappa(G)}{2} -1\Bigr) \ge 0\;,
\end{equation}
so that $\frac{F(G)}{D-1}  - D=0$  if and only $\kappa(G)=2$ and the two connected components are maximally single-trace graphs. In this case, the lower bound that must be satisfied on the $G_i$ themselves is exactly the same, otherwise the bound from Eq.~\eqref{eq:assumpt-2} is strictly larger. But even when the bound given by Eq.~\eqref{eq:assumpt-2} is strictly larger than that of Eq.~\eqref{eq:bound-razvan-cc}, Thm.~\ref{main-theo1} is a significant improvement over Eq.~\eqref{eq:Razvans-condition-for-facto}, in the sense that there is no condition that one should verify for the boundary graphs of the $\mathring G_i$'s. 

\ 

While the bound of Eq.~\eqref{eq:assumpt-2} is in general more constraining than that of Eq.~\eqref{eq:bound-razvan-cc}, it remains non-trivial, as illustrated by the following corollary.
\begin{cor}
Let $D\ge 2$, $p,r \in\bb{N}^*$, and $\delta_1, \ldots, \delta_r \in\bb{N}^*$ such that $\displaystyle\sum_{i=1}^r \delta_i <\frac {D(D-1)}{2}$. If $G_1, \ldots, G_p\in\cG_D$ are compatible and $G'_1, \ldots, G'_r\in\cG_D$ obey $\Delta(G'_i)= \delta_i$ for any $i$, then  $\gG=\{G_1, \ldots, G_p, G'_1, \ldots, G'_r\}$ factorizes over its connected components at large $N$. 
\end{cor}

From the corollary, the disjoint union of arbitrarily many compatible graphs factorizes over its connected components. In Prop.~\ref{prop:families-low-delta}, we highlighted some compatible trace-invariants, and by the use of Thm.~\ref{thm:families-any-delta}, or Thm.~6.4 of Ref.~\cite{Carrozza:2026qcf}, any graph built from these examples by flipping edges is also compatible. 

\

The corollary also applies to a number of non-compatible graphs. For $D=3$, the condition Eq.~\eqref{eq:assumpt-2} implies for instance that if $G_1, G_2$ have $\Delta(G_1) = \Delta(G_2) = 1$, then they factorize at large $N$. As another example, for any $D\geq 2$ we can consider ${D(D-1)}/{2} - 1$ realignment moments, all with $\abs{M_1} = \abs{M_2} = 1$, which will factorize by Prop.~\ref{prop:families-low-delta}).

\paragraph{Second main theorem.} A similar method provides a remarkable improvement to Thm.~\ref{thm:families-tree-like-LO}, as we now detail. We recall that colored graphs having tree-like dominant pairings have been introduced in Def.~\ref{def:DEF}, point \ref{it:HAVEtree}.

\begin{lem}
\label{lem:max-of-kappa}
 Consider  $\gG=\{G_i\}_{i=1, \ldots, p}$ with $G_i\in\cG_D$, and $\widehat G\in\cG_{D+1}^{\conn}(\gG)$. If each $G_i$ factorizes over its connected components at large $N$ and $\widehat G$ is tree-like on $\gG$, then $\widehat G$ saturates  Eq.~\eqref{eq:upper-bound-kappa}, that is, $\kappa(\widehat G) = \kappa(G) - p +1$.
\end{lem}
\proof 
We verify by induction on $p$ that if each $G_i$ factorizes over its connected components at large $N$, tree-like graphs on $\gG$ saturate Eq.~\eqref{eq:upper-bound-kappa}. The reasoning is  very similar to the proof of Lem.~\ref{lem:faces-of-treelike}. If $p=1$, one has $G=G_1$, which by definition satisfies the two-cut property in $\widehat G$: there exists a dominant pairing $\pi_{\rm{opt}}(G_1)$ such that, for every pair of vertices in $\pi_{\rm{opt}}(G_1)$, any edge of color $0$ links two vertices of a same pair. Since $G_1$ factorizes over its connected components at large $N$, every pair of vertices in $\pi_{\rm{opt}}(G_1)$ has both vertices on the same connected component of $G_1$. Therefore: $\kappa(\widehat G) = \kappa(G_1)$. 
 If now $p>1$, consider two edges of color 0 in $\widehat{G}$ attached to the vertices of a dominant pairing and which form a two-cut, flip these two edges, thus obtaining two smaller graphs $\widehat G_a\in \cG_{D+1}^{\conn}(\gG_a)$ and $\widehat G_b\in \cG_{D+1}^{\conn}(\gG_b)$ with $\gG_a\cup \gG_b=\gG$,  which are respectively tree-like on $\gG_a$ and $\gG_b$. We respectively let $G_a$ be the union of the elements of $\gG_a$, and $p_a$ be the number of elements in $\gG_a$, and similarly for $\gG_b$. 
By induction, $\kappa(\widehat G_a)= \kappa(G_a) -p_a+1$, and similarly for $G_b$.  One has $p_a+p_b=p$, $\kappa(G_a) + \kappa(G_b) = \kappa(G)$, and $\kappa(\widehat G) = \kappa(\widehat G_a) + \kappa(\widehat G_b) - 1$, so that $\kappa(\widehat G) = \kappa(G) - p +1$. 
\qed

\ 

\begin{cor}
    \label{cor:max-of-kappa}
Consider  $\gG=\{G_i\}_{i=1, \ldots, p}$ with $G_i\in\cG_D$, and $\widehat G\in\cG_{D+1}^{\conn}(\gG)$.  Then the following upper bound holds:
 \begin{equation}
        \label{eq:upper-bound-F0c}
\mFc(\gG) \le \frac D 2 k(G) + \frac {F(G)} {D-1}  -D(p-1)\;.
 \end{equation}
Furthermore, if each $G_i$ factorizes over its connected components at large $N$ and $\widehat G$ is tree-like on $\gG$, then the bound of Lem.~\ref{lem:Razvans-bound-on-F} coincides with this bound. 
\end{cor}

\proof It is a direct application of Lem.~\ref{lem:Razvans-bound-on-F} and Lem.~\ref{lem:max-of-kappa}. \qed

\ 

\begin{theo}
    \label{main-theo2}
    For any $D\ge 2$ and $p\in\bb{N}^*$, if $G_1, \ldots, G_p\in\cG_D$ are compatible, then $G=G_1\sqcup \cdots \sqcup G_p$ is compatible, and  $\gG=\{G_1, \ldots, G_p\}$ has tree-like dominant pairings. 
\end{theo} 

\proof Consider $\gG$ as in the statement. The compatibility of $G$ follows from the fact that it factorizes over its connected components. Indeed, Eq.~\eqref{eq:mF-equal-parts} yields:
\begin{equation}
    \Delta(G) = \Delta(G_1) + \dots + \Delta(G_p) = 0 \,,
\end{equation}
since the $G_i$'s are compatible.  We now consider $\widehat G$  tree-like on $\gG$ and  show that  $F_0(\widehat G)$ saturates Eq.~\eqref{eq:upper-bound-F0c}. From Eq.~\eqref{eq:Delta_0}, the $G_i$'s being compatible: 
\begin{equation}
\label{eq:sum-maxfaces-comp}
        \sum_{i}\mF(G_i) =    \frac{D}2 k(G) + \frac {F(G)} {D-1} \,.
\end{equation}
Since $\widehat G$ is tree-like on $\gG$,  Lem.~\ref{lem:faces-of-treelike} and Eq.~\eqref{eq:sum-maxfaces-comp} imply that 
   \begin{equation}
           \label{tree-like-compat}
     F_0(\widehat G) = D + \sum_{i=1}^p \bigl( \mF(G_i) - D\bigr) = \frac{D}2 k(G) + \frac {F(G)} {D-1} -  D(p-1). 
 \end{equation}
Tree-like graphs realize the upper bound Eq.~\eqref{eq:upper-bound-F0c} while belonging to $\cG_{D+1}^{\conn}(\gG)$, so they belong to  $\MD{}^{\conn}(\gG)$.\qed

\

Reasoning similarly with tree-like graphs, one shows that if $\gG=\{G_1 , \dots,G_p\}$ satisfies Eq.~\eqref{eq:assumpt-2}, then $\mFc(\gG)$ is close to the bound given by Cor.~\ref{cor:max-of-kappa}, in the sense that:
\begin{equation}
0 \le \frac{D}2 k(G) +  \frac {F(G)} {D-1} -  D(p-1)  - \mFc(\gG)  < D \;. 
\end{equation}
Hence, Thm.~\ref{main-theo2} reveals a striking property of compatible invariants: regardless of the precise structure of the graphs composing $\gG$, the mere fact that $\gG$ is compatible ensures that it exhibits tree-like dominant pairings.

\paragraph{Third main theorem.} Our third main observation is that the existence of tree-like dominant pairings implies large $N$ factorization.
\begin{theo}
\label{main-theo3}
Let $D \geq 2$, $p \in \bb{N}^*$, $G_1,\dots,G_p \in \cG_{D}$ and $\gG=\{G_1, \dots, G_p\}$. If $\gG$ has tree-like dominant pairings, then $\gG$ factorizes at large $N$.  
\end{theo}

\proof
If $\gG$ has tree-like dominant pairings, from Lem.~\ref{lem:tree-like-subsets}, the same is true for all its subsets $\{G_i\}_{i\in I}$, where $I\subset \{1, \ldots, p\}$. From Lem.~\ref{lem:tree-dom-iff-faces}, the necessary and sufficient condition Eq.~\eqref{eq:cond-on-faces-for-facto} for the factorization Eq.~\eqref{eq:large-N-facto} to hold for $\gG$ at large $N$ can therefore be reformulated as: 
\begin{equation}
    \forall \pi \neq 0_p,\qquad D\#(\pi) + \sum_{B\in \pi}\sum_{i\in B} \bigl( \mF(G_i) - D\bigr)  < Dp + \sum_{i=1}^p \bigl( \mF(G_i) - D\bigr)  \;, 
\end{equation}
which simplifies as:
\begin{equation}
  \forall \pi \neq 0_p,\qquad    \#(\pi)  < p\;,
\end{equation}
which is always true.  \qed 

\

The condition required for Thm.~\ref{main-theo3} is stronger than that required for Thm.~\ref{main-theo1}, in the sense that it puts a structural constraint on the graphs in $\MD{}^{\conn}(\gG)$. On the other hand, most of the trace-invariants that have been explicitly studied in the literature on tensor models do satisfy this assumption. The only exception we are aware of is in fact given by the family of maximally single-trace graphs, which was first studied as a whole in Ref.~\cite{Ferrari:2017jgw} (and includes invariants that have been well-studied prior to that work in the context of bilocal melonic theories \cite{Gurau:2010ba, Bonzom:2011zz, Carrozza:2015adg}). In hindsight, it should therefore come as no surprise that the first explicit counterexamples to large $N$ factorization were found in this family (first in Ref.~\cite{Berthold:2026zxk} in a real setting, and now in Fig.~\ref{fig:MST_incomp} in the complex case).

The trace-invariants associated with a colored graph that is given by the disjoint union of the connected graphs listed in Thm.~\ref{thm:families-tree-like-LO} satisfy Thm.~\ref{main-theo3}, \ie they factorize over their connected components. All the discussed examples in this paper that satisfy Thm.~\ref{main-theo1} also satisfy Thm.~\ref{main-theo3} since they all have tree-like dominant pairings. However, regarding our examples, the converse is not true. Indeed, for $D=3$, there exist non-connected planar graphs $G$ such that $\Delta(G)\geq D(D-1)/2$ that factorize over their connected components due to Thm.~\ref{main-theo3}. Moreover, for $D \geq 4$ and $k \in \bb{N}^*$ with $4k \geq D^2 - D + 4$, a disjoint union $G$ of two cyclic graphs, both associated with a subset $M \in \paa{1,\dots,D}$ with $\abs{M} = 2$ and having $2k$ vertices, has $\Delta(G)\geq D(D-1)/2$ but factorizes over its connected components due to Thm.~\ref{main-theo3}.

\ 

Let us focus on the interesting special case of $3$-colored graphs ($D=3$). By Thm.~\ref{main-theo2}, we already know that incompatibility (and even $\Delta \geq 3$) is a necessary condition for a family of $3$-colored graphs to fail to factorize at large $N$. As established by the following corollary, it turns out that non-planarity is also necessary for non-factorization.
\begin{cor} \label{cor:planar}
    Let $p \in \bb{N}^*$, $G_1,\dots,G_p \in \cG_3$ and $\gG = \paa{G_1,\dots,G_p}$. If $\gG$ consists of planar $3$-colored graphs, then, $\gG$ factorizes at large $N$.
\end{cor}
\noindent Since planar graphs have tree-like dominant pairings (see Ref.~\cite{Bonzom:2018btd} and Thm.~\ref{thm:families-tree-like-LO}), Cor.~\ref{cor:planar} is a direct application of Thm.~\ref{main-theo3}.

\section{Application: multipartite entanglement R\'{e}nyi entropies of random quantum states}
\label{subsec:averageApprox}

\paragraph{Context.} Within the theory of multipartite quantum entanglement, trace-invariants can be used to define various quantities such as entanglement monotones (see \eg Refs.~\cite{Gadde:2025csh,Carrozza:2026qcf}), multipartite entropy (see \eg Refs.~\cite{Penington:2022dhr,Gadde:2023zzj,Harper:2024ker,Iizuka:2024pzm,Gadde:2023zni}) or genuine multipartite entropy (see \eg Refs.~\cite{Basak:2024uwc,kmpl-mdbx,Iizuka:2025ioc,Iizuka:2025caq}). Once a specific construction, defined via some function $h$ depending on one or more trace-invariants, is chosen, it is often useful to estimate its typical value for a state drawn uniformly at random according to the Haar measure, especially in large quantum systems (see for instance Refs.~\cite{Dahlsten:2014whd,Singh:2016xzo,Bianchi:2021aui}).

In practice, while the mean value of the trace-invariant itself (or at least its leading contribution) is computable (see Eq.~\eqref{eq:exp-and-asympt-moments}), determining the expectation value of a non-linear function $h$ of the trace-invariant may be more challenging. Nonetheless, in favorable cases, measure concentration phenomena allow to establish asymptotic relations of the form
\begin{equation} \label{eq:ToProofFacto}
    \mean{h\pa{\tr_G\pa{\cT^{(N)},\bar{\cT}^{(N)}}}} \underset{N \to \infty}{\sim} h\pa{\mean{\tr_G\pa{\cT^{(N)},\bar{\cT}^{(N)}}}},
\end{equation}
where $\cT^{(N)}$ denotes a Haar or Gaussian random tensor on a $D$-partite state space of local dimension $N$ (as introduced in Sec.~\ref{sec:Prerequisites}, paragraph ``Trace-invariants and correlations.''). The relation Eq.~\eqref{eq:ToProofFacto} asserts that taking the expectation value and applying the function $h$ are commuting operations to leading order in $1/N$. It is of particular interest to us to assess under which conditions this relation holds for the function $h(\cdot) = -\ln |\cdot|$: indeed, when it does, it allows one to determine an asymptotic expression for the expectation value of the following $\LO$-monotone (see Ref.~\cite{Carrozza:2026qcf}, Example 2.35), defined for a given tensor $S$ as
\begin{equation}
    R_G (S,\bar{S})= \begin{cases} - \ln\abs{\tr_G(S,\bar{S})} &\; \mathrm{if}\; \tr_G(S,\bar{S}) \neq 0 \\ 
    +\infty &\; \mathrm{otherwise\footnotemark} \end{cases} 
\end{equation}
\footnotetext{An example of trace-invariant together with an explicit tensor $S$ for which $\tr_G(S,\bar{S})=0$, was presented in Ref.~\cite{Carrozza:2026qcf} (see Example 2.38).}
In the bipartite setting ($D=2$), for a connected graph $G$ with $k(G) = k$, these quantities reduce to R\'{e}nyi-$k$ entanglement entropies; in the multipartite setting ($D\geq 3$), they can therefore be understood as generalizations of the latter that are indexed by $D$-colored graphs rather than integers.

We will start by exploiting a general concentration phenomenon to bound fluctuations of $\tr_G\pa{\cT^{(N)},\bar{\cT}^{(N)}}$ for any graph $G$ such that the expectation of $G\sqcup \bar{G}$ factorizes.  Next, we will impose further restrictions on  the choice of graph $G$ that will enable us to show that 
\begin{equation}\label{eq:as_equiv_R_G}
    \mean{R_G\pa{\cT^{(N)},\bar{\cT}^{(N)}}} \underset{N\to \infty}{\sim} - \ln\left( \mean{ \tr_G\pa{\cT^{(N)},\bar{\cT}^{(N)}} }\right)\,.
\end{equation}

\paragraph{Concentration phenomenon.} Let $\cT = \pa{\cT^{(N)}}_{N \in \mathbb{N}^*}$, where for every $N \in \mathbb{N}^*$, $\cT^{(N)}$ denotes a Haar random or Gaussian complex $D$-partite tensor of local dimension $N$. We will prove that the deviation of $\tr_G(\cT^{(N)},\bar{\cT}^{(N)})$ from the dominant contribution to the expectation can be made arbitrarily small for sufficiently large $N$.

Our proof will require that $G \in \cG_D$ satisfies the large $N$ factorization condition: 
\begin{equation}\label{eq:factorization_criterion_Haar}
    \mean{\tr_G\pa{\cT^{(N)},\bar{\cT}^{(N)}} \cdot \tr_{\bar G}\pa{\cT^{(N)},\bar{\cT}^{(N)}} } \underset{N \to \infty}{=} \mean{ \tr_G\pa{\cT^{(N)},\bar{\cT}^{(N)}} } \cdot \mean{ \tr_{\bar G}\pa{\cT^{(N)},\bar{\cT}^{(N)}} } \left(1 + O\pa{\frac 1 N}\right)\,.
\end{equation}
More generally, we introduce the following definition, which can be applied to arbitrary sequences of random tensors.
\begin{defi}
Let $S = \pa{S^{(N)}}_{N \in \mathbb{N}^*}$ denote a sequence of random complex $D$-partite tensors and $G \in \cG_D$. We will say that:
\begin{enumerate}
    \item the pair $(G, S)$ -- or, by slight abuse of notation, $\mean{\tr_G \left( S^{(N)}, \bar{S}^{(N)}\right)}$ -- obeys the \emph{large-$N$ Ansatz} if one can find $\mu_G(S) \in \mathbb{C}^*$ and $s_G (S) \in \mathbb{R}$ such that:\footnote{Note that, $|\tr_G (\cdot )|$ being bounded from above by $1$ (see \eg Refs.~\cite{Gadde:2025csh,Carrozza:2026qcf}), one necessarily has $s_G(S)\leq 0$.}
    \begin{equation}\label{eq:Ansatz_large-N}
            \mean{\tr_G\pa{S^{(N)}, \bar{S}^{(N)}}} \underset{N \to \infty}{=} \mu_G(S) N^{s_G(S)} \left( 1 +  O\pa{1/N}\right) \,;
        \end{equation}
    \item $(G, S)$ obeys the \emph{large-$N$ factorization criterion} whenever:
    \begin{equation}\label{eq:factorization_criterion}
    \mean{\abs{\tr_G\pa{S^{(N)}, \bar{S}^{(N)}} }^2 } \underset{N \to \infty}{=} \abs{\mean{\tr_G\pa{S^{(N)}, \bar{S}^{(N)}}}}^2
    \left( 1+ O\left(\frac 1 N\right)\right)\,.
\end{equation}
\end{enumerate}
\end{defi}
If the sequence $S$ obeys the large $N$ Ansatz \eqref{eq:Ansatz_large-N}, the large-$N$ factorization criterion \eqref{eq:factorization_criterion} is equivalent to the following assumptions:
\begin{equation}   
\label{eq:largeNfactoproof}
\mu_{G \sqcup \bar G}(S) = \mu_G(S) \mu_{\bar{G}}(S)=  \abs{\mu_G(S)}^2 \,, \qquad   s_{G\sqcup \bar G}(S) = s_G(S) + s_{\bar G}(S) = 2 s_G(S) \,.
\end{equation}
Furthermore, the Haar or the Gaussian sequence satisfies the large-$N$ Ansatz for any $G \in \cG_D$, and the coefficients $\{\mu_G(\cT)\}_{G\in \cG_D}$ are integers of a combinatorial origin: they count the number of Wick contractions contributing to the leading order.\footnote{More precisely, they are the coefficients defined Eq.~\eqref{eq:mu_def}.} As a result, in this special case, the large-$N$ factorization criterion is obeyed by a graph $G \in \cG_D$ if and only if:
\begin{equation}\label{eq:largeNfacto_Haar}
    \mu_{G \sqcup \bar G}(\cT) = \mu_G(\cT) \mu_{\bar{G}}(\cT)=  \mu_G(\cT)^2 \,, \qquad  \rm{ and } \qquad s_{G\sqcup \bar G}(\cT) = s_G(\cT) + s_{\bar G}(\cT) = 2 s_G(\cT)  \,.
\end{equation}
Concretly, for $G \in \cG_D$, we have
\begin{equation}
    s_G(\cT) = \mF(G) - Dk(G) \,.
\end{equation}

\begin{prop}\label{prop:concentration} 
    Let $S= \pa{S^{(N)}}_{N\in \mathbb{N}^*}$ be a sequence of random $D$-partite tensors and $G \in \cG_D$. If $(G,S)$ obeys both the large-$N$ Ansatz \eqref{eq:Ansatz_large-N} and the large-$N$ factorization criterion \eqref{eq:factorization_criterion} (or equivalently \eqref{eq:largeNfactoproof}), then: for any $\varepsilon > 0$, there exists a constant $N_c > 0$ such that
    \begin{equation}\label{eq:concentration}
      \forall N \in \mathbb{N}^*\,,\qquad  \Prob{ \abs{\frac{\abs{\tr_G\pa{S^{(N)},\bar{S}^{(N)}}}}{\abs{\mu_G(S)}N^{s_G(S)}} - 1}  < \varepsilon}  \geq 1 - \frac{N_c}{N} \,.
    \end{equation} 
In particular, if the invariant $\tr_G$ does not admit zeroes (so that $R_G$ is always real-valued), the typical value of $R_G\pa{S^{(N)},\bar{S}^{(N)}}$ in the large-$N$ regime is:
\begin{equation}
    \abs{s_G (S)} \ln(N)- \ln\abs{\mu_G (S)}\,.
\end{equation}
\end{prop}
The following proof follows the approach of Ref.~\cite[33]{Hayden:2016cfa}, where we emphasize the use of Eq.~\eqref{eq:factorization_criterion}.

\begin{proof}
Fix a graph $G\in\cG_D$ with $2k$ vertices, $N \in \mathbb{N}^*$ and $\varepsilon >0$. Our goal is to bound the fluctuation of $\tr_G\pa{S^{(N)},\bar{S}^{(N)}}$ around the leading-order term of the average. By application of a triangle inequality, we have
    \begin{align}
        \mean{ \abs{\frac{\abs{\tr_G\pa{S^{(N)},\bar{S}^{(N)}}}}{\abs{\mu_G(S)}N^{s_G(S)}} - 1}^2 } &\leq \mean{ \abs{\frac{\tr_G\pa{S^{(N)},\bar{S}^{(N)}}}{\mu_G(S)N^{s_G(S)}} - 1}^2 } \\
        &= \frac{\mean{\abs{\tr_G\pa{S^{(N)},\bar{S}^{(N)}}}}^2}{\abs{\mu_G(S)}^2N^{2s_G(S)}} - 1 - 2 \pa{\mathrm{Re}\frac{\mean{\tr_G\pa{S^{(N)},\bar{S}^{(N)}}}}{\mu_G(S)N^{s_G(S)}} - 1} \,.
    \end{align}
    Noting that the left-hand-side is real and positive, we can apply another triangle inequality to obtain the bound
    \begin{equation} \label{eq:ineqproof}
        0 \leq \mean{ \abs{\frac{\abs{\tr_G\pa{S^{(N)},\bar{S}^{(N)}}}}{\abs{\mu_G(S)}N^{s_G(S)}} - 1}^2 } \leq \abs{\frac{\mean{ \abs{ \tr_G\pa{S^{(N)},\bar{S}^{(N)}} }^2 } }{\abs{\mu_G(S)}^2N^{2s_G(S)}} - 1}+ 2 \abs{\mathrm{Re}\frac{\mean{\tr_G\pa{S^{(N)},\bar{S}^{(N)}}}}{\mu_G(S)N^{s_G(S)}} - 1}\,.
    \end{equation}
The large-$N$ Ansatz applied to $(G, S)$ implies that:
\begin{equation}
    \abs{\mathrm{Re}\frac{\mean{\tr_G\pa{S^{(N)},\bar{S}^{(N)}}}}{\mu_G(S)N^{s_G(S)}} - 1} \underset{N\to \infty}{=} \abs{\mathrm{Re}\frac{\mu_G(S)N^{s_G(S)} \left( 1+ O(1/N)\right)}{\mu_G(S)N^{s_G(S)}} - 1} \underset{N\to \infty}{=} O(1/N)\,.
\end{equation}
Invoking in addition the factorization criterion applied to $(G, S)$ allows one to write
\begin{align}
    \abs{\frac{\mean{ \abs{ \tr_G\pa{S^{(N)},\bar{S}^{(N)}} }^2 }}{\abs{\mu_G(S)}^2N^{2s_G(S)}} - 1} &\underset{N\to \infty}{=} \abs{\frac{\abs{ \mean{ \tr_G\pa{S^{(N)},\bar{S}^{(N)}} } }^2\left( 1+ O(1/N)\right)}{\abs{\mu_G(S)}^2N^{2s_G(S)}} - 1}\\ &\underset{N\to \infty}{=} \abs{\frac{\abs{\mu_G(S)}^2N^{2s_G(S)}
    \left( 1+ O(1/N)\right)}{\abs{\mu_G(S)}^2N^{2s_G(S)}} - 1}\underset{N\to \infty}{=} O(1/N)\,.
\end{align}
As a result, we can find a constant $K >0$ such that
    \begin{equation}
        0 \leq \mean{ \abs{\frac{\abs{\tr_G\pa{S^{(N)},\bar{S}^{(N)}}}}{\abs{\mu_G(S)}N^{s_G(S)}} - 1}^2 } \leq \frac{K}{N}.
    \end{equation}
    Finally, using Markov's inequality, we find that
    \begin{equation}
        \Prob{\abs{\frac{\abs{\tr_G\pa{S^{(N)},\bar{S}^{(N)}}}}{\abs{\mu_G(S)}N^{s_G(S)}} - 1} \geq \varepsilon} \leq \frac{1}{\varepsilon^2} \mean{\abs{\frac{\abs{\tr_G\pa{S^{(N)},\bar{S}^{(N)}}}}{\abs{\mu_G(S)}N^{s_G(S)}} - 1}^2} \leq \frac{K}{N \varepsilon^2 }  = \frac{N_c}{N}\,,
    \end{equation}
    where $N_c = \frac{K}{\varepsilon^2}>0$. Therefore, we obtain
    \begin{equation}
        \Prob{ \abs{\frac{\abs{\tr_G\pa{S^{(N)},\bar{S}^{(N)}}}}{\abs{\mu_G(S)}N^{s_G(S)}} - 1}  < \varepsilon} = 1 - \Prob{\abs{\frac{\abs{\tr_G\pa{S^{(N)},\bar{S}^{(N)}}}}{\abs{\mu_G(S)}N^{s_G(S)}} - 1} \geq \varepsilon} \geq 1 - \frac{N_c}{N}\,,
    \end{equation}
which concludes the proof.    
\end{proof}

With the exception of the example in Fig.~\ref{fig:MST_incomp}, all graphs considered in this paper satisfy the large-$N$ factorization criterion, Eq.~\eqref{eq:factorization_criterion} (or equivalently, Eq.~\eqref{eq:largeNfactoproof}), as guaranteed by Thms.~\ref{main-theo1} and~\ref{main-theo3}. This includes compatible graphs, cyclic graphs, realignment moments, joint realignment moments, and planar graphs.

\paragraph{Expectation value of $R_G$ in the Haar random state.} Let us now introduce a uniform lower-bound on $|\tr_G(\cdot)|$ which, combined with Prop.~\ref{prop:concentration}, will be sufficient to establish Eq.~\eqref{eq:as_equiv_R_G}.
\begin{defi}\label{def:power_law}
    Let $G \in \cG_D$. For any $N \in \mathbb{N}^*$, we introduce: 
    \begin{equation}
        m_G(N)= \min\Bigl\{ \abs{\tr_G \left( S,\bar S \right)} \, \Big\vert \, S \in (\mathbb{C}^N)^{\otimes D} \,, \, \|S\|=1\Bigr\}\,.
    \end{equation}
We will say that $\tr_G$ (or $G$) follows a \emph{power law} if there exists $K > 0$ and $r\geq 0$ such that:
\begin{equation}\label{eq:power_law}
    \forall N \in \mathbb{N}^*\,, \qquad m_G(N) \geq K N^{-r}\,.
\end{equation}
\end{defi}

In particular, with regard to the trace-invariants introduced in the paragraph ``Some examples of colored graphs'' (Sec.~\ref{sec:Prerequisites}), both cyclic graphs and realignment moments follow a power law. For interested readers, certain trace-invariants related to holography or quantum information have also been shown to follow a power-law (see Prop.~6.16 of Ref.~\cite{Carrozza:2026qcf}).

\begin{prop}\label{prop:asymptotic_R_G}   
Let $S=\pa{S^{(N)}}_{N\in \mathbb{N}^*}$ and $G \in \cG_D$ such that: 1) $(G,S)$ obeys both the large-$N$ Ansatz \eqref{eq:Ansatz_large-N} and the large-$N$ factorization criterion \eqref{eq:factorization_criterion} (or equivalently \eqref{eq:largeNfactoproof}); and 2) $\tr_G$ follows a power law. Then:
   \begin{equation}
       \mean{ R_G \pa{S^{(N)},\bar{S}^{(N)}} } \underset{N\to \infty}{=} - \ln \abs{\mean{ \tr_G \pa{S^{(N)},\bar{S}^{(N)}} }}+o(1) \underset{N\to \infty}{=} \abs{s_G(S)} \ln\left(N\right) - \ln\abs{\mu_G(S)}+o(1)\,.
\end{equation}
\end{prop}
\begin{proof}
Let $G \in \cG_D$, $k=k(G)$, and assume the existance of $K>0$, $r\geq 0$ such that Eq.~\eqref{eq:power_law} holds (\ie $\tr_G$ follows a power law). Let us also assume that $(G, S)$ obeys the large-$N$ Ansatz and the large-$N$ factorization criterion.

Let $\delta >0$; we will prove that $\mean{R_G \pa{S^{(N)},\bar{S}^{(N)}}}$ is $\delta$-close from $\abs{s_G(S)} \ln\left(N\right) - \ln\abs{\mu_G(S)}$ for sufficiently large $N$. To this effect, let us fix $\varepsilon >0$ such that $-\ln(1-\varepsilon)$ and $\ln(1+\varepsilon)$ are both in $]0,\delta / 2[$. Given our assumptions, we can apply the concentration result of Prop.~\ref{prop:concentration}: there exists $N_c >0$ such that Eq.~\eqref{eq:concentration} holds. Let us also fix $N \in \mathbb{N}^*$. With probability \emph{at least} $1- \frac{N_c}{N}$, we have
\begin{equation}
    \abs{\frac{\abs{ \tr_G\pa{S^{(N)},\bar{S}^{(N)}} }}{\abs{\mu_G(S)}N^{s_G(S)}} - 1}   < \varepsilon \,,
\end{equation}
implying that
\begin{equation}
    -\ln\left(1+\varepsilon\right) \leq - \ln \abs{\tr_G\pa{S^{(N)},\bar{S}^{(N)}}}- \abs{s_G(S)} \ln\left(N\right) + \ln\abs{\mu_G(S)} \leq -\ln\left(1-\varepsilon\right) \,.
\end{equation}
It follows that:
\begin{equation} \label{eq:asymptotic-bound1}
    \Bigl\vert R_G\pa{S^{(N)},\bar{S}^{(N)}}- \abs{s_G(S)} \ln\left(N\right) + \ln\abs{\mu_G(S)}\Bigr\vert \leq \frac{\delta}{2}\,.
\end{equation}
On the other hand, with probability \emph{at most} $\frac{N_c}{N}$, $\tr_G\pa{S^{(N)},\bar{S}^{(N)}}$ is not necessarily close to its typical value, but we can in this case use the uniform bound Eq.~\eqref{eq:power_law} combined with the triangle inequality to write:
\begin{align}\label{eq:asymptotic-bound2}
\Bigl\vert R_G\pa{S^{(N)},\bar{S}^{(N)}}- \abs{s_G(S)} \ln\left(N\right) + \ln\abs{\mu_G(S)}\Bigr\vert &\leq - \ln \abs{\tr_G\pa{S^{(N)},\bar{S}^{(N)}}}+ \abs{s_G(S)} \ln\left(N\right) - \ln\abs{\mu_G(S)} \\
&\leq a \ln(N)+ b
\,,
\end{align}
where $a= r+\abs{s_G(S)}$ and $b= \abs{\ln(K)}+ \abs{\ln\abs{\mu_G(S)}}$ are strictly positive constants. By linearity of the expectation value $\mean{\;\cdot\;}$, we can combine the bounds Eq.~\eqref{eq:asymptotic-bound1} (which is valid with probability at most $1$) and Eq.~\eqref{eq:asymptotic-bound2} (which is valid with probability at most $N_c/N$), to obtain
\begin{equation}
    \Bigl\vert \mean{R_G\pa{S^{(N)},\bar{S}^{(N)}}}- \abs{s_G(S)} \ln\left(N\right) + \ln\abs{\mu_G(S)}\Bigr\vert \leq \frac{\delta}{2} + \frac{N_c }{N} \left(a\ln(N) +b\right)\,.
\end{equation}
Since the second term on the right-hand side converges to $0$ in the $N\to +\infty$ limit, it will be bounded by $\delta / 2$ for $N$ sufficiently large. Hence, we can find $N_0\in \mathbb{N}^*$ such that: 
\begin{equation}
\forall N \geq N_0\,, \qquad     \Bigl\vert \mean{R_G\pa{S^{(N)},\bar{S}^{(N)}}}- \abs{s_G(S)} \ln\left(N\right) + \ln\abs{\mu_G(S)}\Bigr\vert \leq \frac{\delta}{2} +\frac{\delta}{2} = \delta\,.
\end{equation}
This concludes the proof that
\begin{equation}
\mean{R_G\pa{S^{(N)},\bar{S}^{(N)}}}   \underset{N\to \infty}{=} \abs{s_G(S)} \ln\left(N\right) - \ln\abs{\mu_G(S)}+o(1)\,,    
\end{equation}
and the right-hand side is clearly equal to $- \ln \abs{\langle \tr_G\pa{S^{(N)},\bar{S}^{(N)}}\rangle}+o(1)$, so we are done.
\end{proof}
\begin{cor}  \label{cor:Lit_inv_and_exchange} Let $\cT = \pa{\cT^{(N)}}_{N \in \mathbb{N}^*}$, where for every $N \in \mathbb{N}^*$, $\cT^{(N)}$ denotes the Haar or Gaussian complex $D$-partite random tensor of local dimension $N$. If $G$ is a $D$-colored graph obeying the large-$N$ factorization criterion \eqref{eq:factorization_criterion_Haar} and the power law bound \eqref{eq:power_law}, then the following asymptotic relation holds:
    \begin{equation}\label{eq:asymptotic-RG_Haar}
        \mean{R_G \pa{\cT^{(N)},\bar{\cT}^{(N)}}} \underset{N\to \infty}{=} - \ln \mean{\tr_G \left( \cT^{(N)},\bar{\cT}^{(N)} \right)} +o(1) \underset{N\to \infty}{=} \pac{Dk(G) - \mF(G)} \ln\left(N\right) - \ln\left(\mu_G\right)+o(1)\,.
    \end{equation}
\end{cor} 
\begin{proof}
This follows straightforwardly from the concentration phenomenon of Prop.~\ref{prop:concentration}, Prop.~\ref{prop:asymptotic_R_G}, and the fact that $\mean{\tr_G \left( \cT^{(N)},\bar{\cT}^{(N)} \right)}$ is positive for any $G\in \cG_D$.
\end{proof}

Cyclic graphs and realignment moments with $\abs{M_3} \geq \abs{M_1},\abs{M_2}$ follow a power law (see Ref.~\cite{Carrozza:2026qcf}). Furthermore, both exhibit tree-like dominant pairings (see Thm.~\ref{thm:families-tree-like-LO}), a fact which, by Thm.~\ref{main-theo3}, leads to their large-$N$ factorization. As such, they serve as concrete examples of graphs satisfying the asymptotic relation in Eq.~\eqref{eq:asymptotic-RG_Haar} (for more examples, see Ref.~\cite[87]{Carrozza:2026qcf}).

\ 

In conclusion, we have seen so far in this section that, considering a graph $G$ whose associated trace-invariant $\tr_G$ admits no zeroes (so that $R_G$ is real-valued), the \emph{typical value} of $R_G \left( \cT^{(N)},\bar{\cT}^{(N)} \right)$ in the Haar random state is easily determined whenever $G$ obeys the large $N$ factorization property \eqref{eq:factorization_criterion_Haar}. If $G$ obeys in addition a power law (Def.~\ref{def:power_law}), then an asymptotic expression for the expectation value $\mean{R_G \left( \cT^{(N)},\bar{\cT}^{(N)} \right)}$ can be straightforwardly determined by commuting the average and the logarithm (in other words, the "annealed" and "quenched" computations yield the same result). In the bipartite setting ($D=2$), it turns out that those two conditions are always verified, hence expectation values of R\'{e}nyi-$k$ entanglement entropies (with integer $k$) can be easily evaluated in the Haar random state (which is well-known). By contrast, neither of those conditions on $G$ automatically hold in the multipartite setting ($D\geq 3$), and neither of them implies the other. When they do not, determining the average of $\mean{R_G \left( \cT^{(N)},\bar{\cT}^{(N)} \right)}$ (or even its typical value) becomes more challenging in the sense that we should expect the result of the quenched calculation to fail to accurately reproduce $\mean{R_G \left( \cT^{(N)},\bar{\cT}^{(N)} \right)}$. Moreover, properly defining $\mean{R_G \left( \cT^{(N)},\bar{\cT}^{(N)} \right)}$ may itself be challenging whenever $\tr_G$ admits zeroes. We briefly consider this problem in the next paragraph, focusing on the particularly interesting case of non-factorizing maximally single-trace invariants. 

\paragraph{Non-factorizing maximally single-trace invariants.} Consider a maximally single-trace graph $H$ that fails to obey the large $N$ factorization property \eqref{eq:factorization_criterion_Haar}, such as \eg the graph $H$ of Fig.~\ref{fig:MST_incomp}. We also expect such an invariant to admit zeroes,\footnote{While a general proof is currently missing, Ref.~\cite{Carrozza:2026qcf} provides one example of maximally single-trace invariant admitting zeroes.} and in particular to fail to obey a power law (Def.~\ref{def:power_law}). In fact, the existence of zeroes of $\tr_H$ makes the very definition of $\mean{R_H \left( \cT^{(N)},\bar{\cT}^{(N)} \right)}$ somewhat questionable, since $R_H$ formally evaluates to $+\infty$ on any such zero.\footnote{Note that, at fixed $N$, the locus of points where $\tr_H$ vanishes is by definition an algebraic set; it is thus of measure $0$ with respect to the Haar measure. Hence, the conventional value we assign to $R_H$ at those zeroes is irrelevant for the calculation discussed here; the issue is that $R_H$ reaches arbitrarily large positive values in neigborhoods of those zeroes that contribute non-trivially to the Haar measure.} To resolve this difficulty, we will introduce a regulator to suppress the contributions of large values of $R_H$, and study how this regulator could possibly be removed in a second step. At fixed $N$, a natural way of implementing such a strategy would be to introduce a constant $\Lambda > 0$ and compute the average of the real random variable $\min\left(R_H\left( \cT^{(N)},\bar{\cT}^{(N)} \right),\ln(\Lambda)\right)$. However, since we are interested in taking a scaling limit, we will make the cut-off explicitly dependent on $N$. To fix this dependency, it is instructive to first look at the quenched version of this calculation. Let $G= H \sqcup \bar H$. Since $\tr_G (\cdot) = \abs{\tr_H(\cdot)}^2$, we have $R_H= \frac{1}{2}R_G$. We can thus define a \emph{quenched average} of $R_H$ as:
\begin{equation}
    \mean{R_H\left( \cT^{(N)},\bar{\cT}^{(N)} \right)}_{\mathrm{quenched}} \eqdef - \frac{1}{2} \ln \mean{\exp\left( R_G\left( \cT^{(N)},\bar{\cT}^{(N)} \right) \right) } = - \frac{1}{2} \ln \mean{\tr_G \left( \cT^{(N)},\bar{\cT}^{(N)} \right) }\,. 
\end{equation}
The right-hand side is well-defined since the expectation value $\mean{\tr_G \left( \cT^{(N)},\bar{\cT}^{(N)} \right) }$ is strictly positive and, by Prop.~\ref{prop:all-cum-LO}, it is easily evaluated to:
\begin{equation}
    \mean{R_H\left( \cT^{(N)},\bar{\cT}^{(N)} \right)}_{\mathrm{quenched}} =  \frac{D k(H)}{2} \ln(N)-\frac{1}{2}\ln(\mu^{\conn}_G) + o(1)\,. 
\end{equation}
The previous expression suggests to work with the following regularized version of $R_H\left( \cT^{(N)},\bar{\cT}^{(N)} \right)$ in the annealed calculation: for any $\Lambda > 0$ and $N\in \mathbb{N}^*$, we define
\begin{equation}
    R_H^{(\Lambda,N)}\left( \cT^{(N)},\bar{\cT}^{(N)} \right) \eqdef \min\left(R_H\left( \cT^{(N)},\bar{\cT}^{(N)} \right),\frac{D k(H)}{2} \ln(N)+\ln(\Lambda)\right)\,.
\end{equation}
It is straightforward to evaluate its expectation value explicitly if we know the asymptotic distribution of the random variable $N^{Dk(H)}\tr_{G}\left( \cT^{(N)},\bar{\cT}^{(N)} \right)$. How to determine this distribution for an arbitrary non-factorizing maximally single-trace invariant has been outlined in and around Prop.~\ref{prop:all-cum-LO-distr}. Let us restrict our attention to the special cases discussed in detail in that proposition. Namely, we now make the further assumptions that $\mF(H)< \frac D 2 k(H)$ and: either $\mF^{\conn}(H\sqcup H)=\mF^{\conn}(H\sqcup \bar H)$, or  $\mF^{\conn}(H\sqcup H)=\mF^{\conn}(H\sqcup \bar H)$ and $\mu^{\conn}_G = \mu^{\conn}_{H\sqcup H}$. According to Prop.~\ref{prop:all-cum-LO-distr}, $N^{Dk(H)}\tr_{G}\left( \cT^{(N)},\bar{\cT}^{(N)} \right)$ converges in distribution to an exponential or a Gamma random variable in the limit $N \to +\infty$. Let us denote by $\rho:\mathbb{R}_+^* \to \mathbb{R}_+$ the probability density of this limit distribution on $\mathbb{R}_+^*$ (note that it is continuous, non-zero everywhere, and decays at least exponentially towards $+\infty$). We then find that 
\begin{equation}
    \mean{R_H^{(\Lambda,N)}\left( \cT^{(N)},\bar{\cT}^{(N)} \right) } \underset{N \to +\infty}{=} \alpha_\Lambda \ln(N) + \beta_\Lambda + o(1)
\end{equation}
with
\begin{align}
\alpha_\Lambda &= \frac{Dk(H)}{2} \left( 1 + \int_{0}^{\Lambda^{-2}} \rho(x) \,\extd x \right) \underset{\Lambda \to +\infty }{\rightarrow} \frac{Dk(H)}{2} \,,\\
    \beta_\Lambda &= - \frac{1}{2} \int_{\Lambda^{-2}}^{+\infty} \rho(x) \ln(x) \,\extd x +\ln(\Lambda)  \int_{0}^{\Lambda^{-2}} \rho(x) \, \extd x \underset{\Lambda \to +\infty }{\rightarrow} - \frac{1}{2} \int_{0}^{+\infty} \rho(x) \ln(x) \, \extd x\,.
\end{align}
Assuming that the large $N$ and large $\Lambda$ limits commute (which we will not attempt to rigorously justify here), we conclude that the regularization can be removed, yielding the annealed result:
\begin{equation}
    \mean{R_H\left( \cT^{(N)},\bar{\cT}^{(N)} \right) } \underset{N \to +\infty}{=} \frac{Dk(H)}{2} \ln(N)  - \frac{1}{2} \int_{0}^{+\infty} \rho(x) \ln(x) \,\extd x + o(1)
\end{equation}
Given that
\begin{equation}
    \int_{0}^{+\infty} \rho(x) \ln(x) \,\extd x \neq \ln(\mu_G^{\conn})
\end{equation}
in both cases covered by Prop.~\ref{prop:all-cum-LO-distr},\footnote{Indeed, one finds $\int_0^{+\infty} \rho(x) \ln(x) \,\extd x = \ln \mu_G^{\conn} - \gamma$ if $\rho$ is an exponential distribution (case $\mF^{\conn}(H\sqcup H)<\mF^{\conn}(H\sqcup \bar H)$) and $\int_0^{+\infty} \rho(x) \ln (x)\, \extd x = \ln \mu_G^{\conn} - \gamma - 2 \ln 2$ if $\rho$ is a Gamma distribution (case $\mF^{\conn}(H\sqcup H)=\mF^{\conn}(H\sqcup \bar H)$ and $\mu^{\conn}_G = \mu^{\conn}_{H\sqcup H}$).} we thus conclude that the quenched and annealed calculations are inequivalent in the large $N$ limit, with the difference showing up at first subleading order: in other words, we have
\begin{equation}
    \mean{R_H\left( \cT^{(N)},\bar{\cT}^{(N)} \right) } \underset{N \to +\infty}{=} \mean{R_H\left( \cT^{(N)},\bar{\cT}^{(N)} \right) }_{\mathrm{quenched}} + o\left(\ln(N)\right)\,,
\end{equation}
but
\begin{equation}
    \mean{R_H\left( \cT^{(N)},\bar{\cT}^{(N)} \right) } \underset{N \to +\infty}{\neq} \mean{R_H\left( \cT^{(N)},\bar{\cT}^{(N)} \right) }_{\mathrm{quenched}} + o(1)\,.
\end{equation}

\addcontentsline{toc}{section}{References}

\printbibliography

\end{document}

%% file: Packages.tex
\usepackage{graphicx} 
\usepackage[left=2cm,right=2cm,top=2cm,bottom=2cm]{geometry} 
\usepackage{xcolor} 
\usepackage{enumitem}
\usepackage{pifont}  
\usepackage{comment} 
\usepackage{float} 
\usepackage{multicol} 

\usepackage{authblk}

\usepackage{soul}

\usepackage{mathtools} 
\usepackage{esvect} 
\usepackage{stackrel,amssymb} 
\usepackage{amsmath} 
\usepackage{amsthm} 
\usepackage{mathrsfs} 
\usepackage{cancel} 
\usepackage{mathabx} 
\usepackage{stackengine} 
\usepackage{tablefootnote}
\usepackage{caption}

\usepackage{pgfplots}
\pgfplotsset{compat=1.18}

\usepackage{slashbox} 
\usepackage{diagbox} 
\usepackage{multirow} 

\usepackage{graphicx}      
\usepackage{eso-pic}       
\graphicspath{ {./images/} }

 \usepackage[
     backend=biber,
     style=numeric-comp,
     sorting=none,
   ]{biblatex}
\usepackage{ifthen}
\usepackage{stackengine} 
\usepackage{scalerel}

\usepackage{listings} 

\definecolor{codegreen}{rgb}{0,0.6,0}
\definecolor{codegray}{rgb}{0.5,0.5,0.5}
\definecolor{codepurple}{rgb}{0.58,0,0.82}
\definecolor{backcolour}{rgb}{0.95,0.95,0.92}

\lstdefinestyle{mystyle}{
    backgroundcolor=\color{backcolour},   
    commentstyle=\color{codegreen},
    keywordstyle=\color{magenta},
    numberstyle=\tiny\color{codegray},
    stringstyle=\color{codepurple},
    basicstyle=\ttfamily\footnotesize,
    breakatwhitespace=false,         
    breaklines=true,                 
    captionpos=b,                    
    keepspaces=true,                 
    numbers=left,                    
    numbersep=5pt,                  
    showspaces=false,                
    showstringspaces=false,
    showtabs=false,                  
    tabsize=2
}

\usepackage{hyperref}
\hypersetup{
    colorlinks=true,
    linkcolor=black,
    citecolor=black,      
    urlcolor=black,
    pdftitle={Tensorial invariants for the classification of multipartite entanglement},
    pdfpagemode=FullScreen,
}

\usepackage{lscape}
\usepackage{graphicx}

%% file: Commandes.tex
\makeatletter 

\newcommand*\Neginternal[3]{\mathpalette\Neg@{{#1}{#2}{#3}}}
\newcommand*\Neg@[2]{\Neg@@{#1}#2}
\newcommand*\Neg@@[4]{%
  \mathrel{\ooalign{%
    $\m@th#1#4$\cr
    \hidewidth$\m@th#3{#1}\mkern\muexpr#2*2$\hidewidth\cr
  }}%
}

\newcommand*\negslash[1]{\m@th#1\not\mathrel{\phantom{=}}}
\newcommand*\snegslash[1]{\rotatebox[origin=c]{60}{$\m@th#1-$}}
\newcommand*\ssnegslash[1]{\rotatebox[origin=c]{60}{$\m@th#1{\dabar@}\mkern-7mu{\dabar@}$}}
\newcommand*\sssnegslash[1]{\rotatebox[origin=c]{60}{$\m@th#1\dabar@$}}
\makeatother  

\newcommand{\lat}[1]{\textsl{#1}}
\newcommand{\apriori}{\lat{a~priori} }
\newcommand{\ie}{\lat{i.e.}~}

\newcommand{\eg}{\lat{e.g.}~}
 
\renewcommand{\rm}[1]{\text{#1}}

\newcommand{\e}[1]{\rm{\textup{e}}^{#1}}           

\renewcommand{\d}{\rm{\textup{d}}} 		      	

\newcommand{\pa}[1]{\left( #1 \right)}
\newcommand{\pac}[1]{\left[ #1 \right]}
\newcommand{\paa}[1]{\left\{ #1 \right\}}

\DeclareMathOperator{\tr}{Tr}          

\newcommand{\mF}{\cal{F}_0}
\newcommand{\mFc}{\cal{F}_0^{\conn}}
\newcommand{\gG}{\mathbf{G}}
\newcommand{\gH}{\mathbf{H}}

\DeclareMathOperator{\id}{id}

\renewcommand{\vec}[1]{\vv{#1}} 				
\renewcommand{\tilde}[1]{\widetilde{#1}}

\newcommand{\norm}[1]{\left\lVert#1\right\rVert}

\newcommand{\ket}[1]{\left\vert #1 \right\rangle} 			
\newcommand{\mean}[1]{\left\langle #1 \right\rangle}

\newcommand{\abs}[1]{\left\lvert #1 \right\rvert}

\renewcommand{\bar}[1]{\overline{#1}}

\newcommand\eqdef{\mathrel{\overset{\makebox[0pt]{\mbox{\normalfont\tiny\sffamily def}}}{=}}}

\newcommand{\bb}[1]{\mathbb{#1}}

\newcommand{\Prob}[1]{\bb{P}\pa{#1}}

\newcommand{\MD}[1]{\cal{M}}

\theoremstyle{plain}
\newtheorem{theo}{Theorem}[section]
\newtheorem{prop}[theo]{Proposition}
\newtheorem{lem}[theo]{Lemma}
\newtheorem{rem}[theo]{Remark}
\newtheorem{ex}[theo]{Example}
\newtheorem{defi}[theo]{Definition}
\newtheorem{cor}[theo]{Corollary}

\renewcommand{\cal}[1]{\mathcal{#1}}

\definecolor{cadmiumred}{rgb}{0.89, 0.0, 0.13}
\newboolean{ifprintcolor}	

\ifthenelse{\boolean{ifprintcolor}}{}{}

\newcommand{\LU}{\mathsf{LU}}
\newcommand{\LO}{\mathsf{LO}}

\newcommand{\cZ}{\mathcal{Z}}

\newcommand{\cT}{\cal{T}}
\newcommand{\I}{%
  \textup{\uppercase\expandafter{\romannumeral 1}}%
}
\newcommand{\II}{%
  \textup{\uppercase\expandafter{\romannumeral 2}}%
}

\DeclareMathOperator{\conn}{c}

\newcommand\restr[2]{{
  \left.\kern-\nulldelimiterspace 
  #1 
  \vphantom{\big|} 
  \right|_{#2} 
  }}
\def\extd{\mathrm {d}}
\newcommand{\cG}{\mathcal{G}}